\documentclass[12pt,reqno]{amsart}
\usepackage{harvard,verbatim}
\usepackage{graphicx}
\usepackage{hyperref}
\usepackage{bm}
\usepackage{amsmath}
\numberwithin{equation}{section}
\numberwithin{figure}{section}
\numberwithin{table}{section}
\makeatletter \def\section{\@startsection{section}{1} {\z@}{-2.5ex
plus -0.5ex minus -0.1ex}{0.5ex plus 0.1ex}{\large\bfseries}}
\def\subsection{\@startsection{subsection}{2} {\z@}{-2.25ex plus
-0.3ex minus -0.2ex}{0.05ex plus 0.05ex}{\normalsize\bfseries}}
\def\subsubsection{\@startsection{subsubsection}{3} {\z@}{-2.25ex plus
-0.3ex minus -0.2ex}{0.05ex plus 0.05ex}{\normalsize\bfseries\em}}
\makeatother
\usepackage{amsmath,amsfonts}
\usepackage{amsmath,amstext}
\usepackage{amsfonts,amssymb,graphics,xspace,endnotes}
\usepackage{epsfig}
\usepackage{rotating}
\numberwithin{equation}{section}

\renewcommand{\cite}{\citeasnoun}

\newcommand{\MM}{\mathbb{M}}

\newcommand{\RR}{\mathbb{R}}

\newtheorem{proposition}{PROPOSITION}

\newcommand{\ba}{\begin{array}}
\newcommand{\ea}{\end{array}}
\newcommand{\bs}{\begin{align}\begin{split}\nonumber}
\newcommand{\bsnumber}{\begin{align}\begin{split}}
\newcommand{\es}{\end{split}\end{align}}

\newcommand{\be}{\begin{eqnarray}}
\newcommand{\ee}{\end{eqnarray}}
\renewcommand{\(}{\left(}
\renewcommand{\)}{\right)}

\newtheorem{assumption}{ASSUMPTION}
\newtheorem{theorem}{THEOREM}
\newtheorem{lemma}{LEMMA}

\newcommand{\EE}{\mathbb{E}}

\theoremstyle{definition}

\usepackage{undertilde}

\renewenvironment{description}[1][0pt]
  {\list{}{\labelwidth=0pt \leftmargin=#1
   }}
  {\endlist}

\usepackage{color}

\begin{document}

\thispagestyle{empty}
\begin{center}
\renewcommand{\thefootnote}{\fnsymbol{footnote}}
 \textbf{\large  A Panel Quantile Approach to Attrition Bias in Big Data: \\ Evidence from a Randomized Experiment\footnote{The authors would like to thank Ju Hyun Kim, Jerry Reiter, Nancy Rose, and Jeffrey Wooldridge for comments on a previous draft as well as conference participants at the Conference in Honor of Jerry Hausman, Wequassett Resort, October 1-3, 2015. We are also grateful to seminar participants at Duke University, University of Illinois at Urbana-Champaign, the 2016 Latin American Econometric Society meeting, and the 26th Annual Meeting of the Midwest Econometrics Group. Danton Noriega-Goodwin provided excellent research assistance.}
 ~\\}
 \vspace{0.2in} Matthew Harding\footnote{Department of Economics and Department of Statistics, University of California - Irvine, 3207 Social Science Plaza B, Irvine, CA 92697; Phone (949) 824 1511; Email: \texttt{harding1@uci.edu}; \url{www.DeepDataLab.org}} and Carlos Lamarche\footnote{Department of Economics, University of Kentucky, 223G Gatton College of Business and Economics, Lexington, KY 40506-0034; Phone: (859) 257 3371;  Email: \texttt{clamarche@uky.edu}}
\end{center}
\begin{center}
\ \today  \\
\end{center}

\vspace{.3in} \hrule   \noindent
\begin{center}
\textbf{Abstract}
\end{center}
\baselineskip=.88\baselineskip  {\small

This paper introduces a quantile regression estimator for panel data models with individual heterogeneity and attrition. The method is motivated by the fact that attrition bias is often encountered in Big Data applications. For example, many users sign-up for the latest program but few remain active users several months later, making the evaluation of such interventions inherently very challenging. Building on earlier work by Hausman and Wise (1979), we provide a simple identification strategy that leads to a two-step estimation procedure. In the first step, the coefficients of interest in the selection equation are consistently estimated using parametric or nonparametric methods. In the second step, standard panel quantile methods are employed on a subset of weighted observations. The estimator is computationally easy to implement in Big Data applications with a large number of subjects.  We investigate the conditions under which the parameter estimator is asymptotically Gaussian and we carry out a series of Monte Carlo simulations to investigate the finite sample properties of the estimator. Lastly, using a simulation exercise, we apply the method to the evaluation of a recent Time-of-Day electricity pricing experiment inspired by the work of Aigner and Hausman (1980).

\vspace{.1in}

\baselineskip=1.1\baselineskip
 \vspace{.1in} \hrule \vspace{.2in}
\setcounter{page}{1}
\setcounter{footnote}{0}

\clearpage

\section{Introduction}

Panel data, or multiple observations of the same unit over time, provides social scientists with the opportunity of examining complex relationships and addressing a wide range of estimation issues that could not be addressed using only cross-sectional data (see, e.g., Hsiao 2014\nocite{cH14}, Baltagi 2013\nocite{bB13}, Arellano and Honor\'e 2000\nocite{Arellano20013229}, among others). At the same time, the use of experimental data allows social scientists to identify and consistently estimate treatment effects using a random sample of subjects. However, as originally pointed out by Hausman and Wise (1979)\nocite{dA80}, data collection over time in an experimental setting raises the issue of ``non-ignorable'' missing data, or attrition. In this paper, we investigate non-random attrition in large randomized field trials arising from the increased availability of Big Data.

It is known that several methods are available to address selection and attrition in both cross-sectional and panel data models. Numerous papers have proposed methods for estimating conditional mean models or average treatment effects while dealing with missing observations (e.g., Hausman and Wise (1979), Ridder (1992), Kyriazidou (1997)\nocite{eKyriazidou97}, Hirano, Imbens, Ridder, and Rubin (2001), Nevo (2003), Das (2004)\nocite{Das2004}, Bhattacharya (2008), among others). Quantile regression is becoming increasingly popular in applied microeconomic research and offers a semiparametric alternative approach to standard methods. Although there is a recent number of papers that investigate estimation of a panel quantile regression model (Koenker (2004)\nocite{rK04}, Abrevaya and Dahl (2008)\nocite{Abre05}, Lamarche (2010)\nocite{cL06}, Canay (2011)\nocite{iC11}, Rosen (2012)\nocite{aR10}, Galvao, Lamarche, and Lima (2013)\nocite{cL14}, Chernozhukov, Fern\'{a}ndez-Val, Hahn, and Newey (2013)\nocite{vCwN13}, Harding and Lamarche (2014)\nocite{mH13}, Chernozhukov, Fern\'{a}ndez-Val, Hoderlein, Holzmann, and Newey (2015)\nocite{Chernozhukov2015378}, among others), the literature deals exclusively with the case of no missing data or it implicitly assumes random attrition in the case of unbalanced panels. 

This paper proposes a quantile regression estimator for panel data when units do drop out of the sample. We allow the missing data process to occur over time after an initial random sample of subjects are assigned into either treatment or control groups. To the best of our knowledge, the only papers that investigate this issue are Lipsitz et al. (1997) and Maitra and Vahid (2006)\nocite{pM06}. Our paper however makes two contributions relative to the existing literature. First, we propose a method to estimate a model with individual unobserved heterogeneity that can be a source of attrition. Second, the proposed estimator handles attrition that can arise from both selection on observables and selection on unobservables under a time-homogeneity condition on the missing data process. We illustrate the use of the approach considering a time-of-use electricity pricing where the condition is likely to be satisfied by the availability of Big Data of households' panels. We adjust for selection bias by using additional samples to estimate a propensity score to weight observations.

Although there is a history of contributions in quantile regression that uses weights (see Koenker  (2005, \S5.3); see also, Abadie, Angrist and Imbens (2002)\nocite{Alberto02}, Portnoy (2003)\nocite{sP03}, Wang and Wang (2009)\nocite{jW09}, among others), they have been employed in cross-sectional data under different models and assumptions. For instance, Wang and Wang (2009) construct a weighted quantile objective function with the idea of redistributing the mass of observations in a censored quantile problem. Lipsitz et al. (1997)\nocite{sLipsitz1997} and Maitra and Vahid (2006) propose a weighting scheme for longitudinal data but their estimating equations would lead to inconsistent and inefficient results in the attrition model of Hausman and Wise (1979). When additional data are available, it is possible to adjust the selection bias as shown in Ridder (1992)\nocite{gRidder92}, Hirano, Imbens, Ridder, and Rubin (2001)\nocite{kHirano01}, Bhattacharya (2008)\nocite{Bhattacharya2008}, and Deng, Hillygus, Reiter, Si, and Zheng (2013)\nocite{deng2013}. This paper illustrates that it is possible to correct the moment condition corresponding to a conditional quantile panel data problem to avoid biased and inconsistent results in the spirit of Nevo (2003).   

The next section introduces the model for missing data and the proposed estimator. It also shows the asymptotic properties of the estimator. In Section 3, we investigate the small sample performance of the proposed approach considering the cases of completely ignorable and non-ignorable missing data patters. Section 4 illustrates the theory and provides practical guidelines from an application of the method to a simulation exercise using a large randomized trial. We investigate the impact of considering different simulated models of attrition on the performance of several panel quantile methods. We explore an application of a recent Time-of-Day electricity pricing and estimate the effect of ``smart'' (communicating) technologies on households' savings from electricity consumption. Section 5 concludes.  

\section{The Model and Proposed Estimator}
\subsection{Background}
Let $Y_{it} \in \RR$ denote a potentially unobserved $t$-th response of the $i$-th individual. The model for $Y_{it}$ for $i=1,\hdots,N$ and $t=1,\hdots,T$ is given by,
\begin{equation}
Y_{it} = \bm{d}_{it}' \bm{\delta} + \bm{x}_{it}' \bm{\beta} + \alpha_i + u_{it}, \label{eq:model}
\end{equation}
where $\bm{d}_{it}$ is a $p_d$-dimensional vector of variables indicating whether the unit is under treatment and whose support is $\mathcal{D} \subseteq \RR^{p_d}$, $\bm{x}_{it}$ is a $p_x$-dimensional vector of exogenous independent variables with support $\mathcal{X} \subseteq \RR^{p_x}$, $\alpha_i$ is a scalar unobserved time-invariant individual effect and $u_{it}$ is an error term. It is assumed that $Y_{it}$ is observed at $t = 1$ for all $i$ and $Y_{it}$ might not be observed at $t > 1$. Let the variable $s_{it}$ indicate whether the $t$-th response of the $i$-th individual is missing. We define $s_{it} = 1$ if and only if the response variable $Y_{it}$ is observed, and 0 otherwise. It is assumed that $\{ (\bm{d}_{it}',\bm{x}_{it}') \}$ are available for all $(i,t)$.    

Under the assumption of no missing data, $s_{it}=1$ for all $(i,t)$, a quantile regression model for equation $\eqref{eq:model}$ can be written as,
\begin{equation}
Q_{Y_{it}}(\tau | \bm{d}_{it},\bm{x}_{it},\alpha_i) = \bm{d}_{it}' \bm{\delta}(\tau) + \bm{x}_{it}' \bm{\beta}(\tau)  + \alpha_i(\tau),  \label{eq:qrmodel} 
\end{equation}
where $\tau$ is a quantile in the interval $(0,1)$ and the conditional quantile function $Q_{Y_{it}}(\tau | \bm{d}_{it},\bm{x}_{it},\alpha_i) = \inf\{ y : P(Y_{it} < y | \bm{d}_{it},\bm{x}_{it},\alpha_i) \geq \tau \}$. The parameter of interest is the quantile specific treatment effect, $\bm{\delta}(\tau)$, and $\alpha _{i}(\tau)$ is a quantile-specific individual effect capturing unobserved and observed time-invariant heterogeneity that was not adequately controlled by the independent variables in model \eqref{eq:model}. The model assumes that observations arise from location-scale shift family of continuous distributions and it can be considered to be semiparametric since the functional form of the conditional distribution of $Y_{it}$ given $(\bm{d}_{it}',\bm{x}_{it}',\alpha_{i})$ is left unspecified. 

When there is no missing data, the model \eqref{eq:qrmodel} can be consistently estimated under $N$ and $T$ tending to infinity (e.g., Koenker (2004)\nocite{rK04}, Kato, Galvao and Montes-Rojas (2012)\nocite{kK12}) by finding the minimizer of,
\begin{equation}
Q_{NT}(\bm{\delta},\bm{\beta},\bm{\alpha}) = \frac{1}{NT} \sum_{i=1}^N \sum_{t=1}^T  \rho_{\tau} ( Y_{it} - \bm{d}_{it}' \bm{\delta} - \bm{x}_{it}' \bm{\beta} - \bm{z}_{i}' \bm{\alpha}), \label{eq:objectivef} 
\end{equation}
where $\rho_{\tau}  = u (\tau - I(u < 0 ))$ is the standard quantile regression check function (Koenker (2005)\nocite{rK05}), $\bm{\alpha}(\tau) = (\alpha_1(\tau),\hdots,\alpha_N(\tau))'$ is a vector of individual effects, and $\bm{z}_{i} = (0, \hdots, 1, \hdots, 0)'$ is an $N$-dimensional ``incidence" vector. The minimizer of \eqref{eq:objectivef} is also the solution of the following estimating equation:
\begin{equation}
M_{NT}(\bm{\delta},\bm{\beta},\bm{\alpha}) = - \frac{1}{NT} \sum_{i=1}^N \sum_{t=1}^T  (\bm{d}_{it}',\bm{x}_{it}',\bm{z}_{i}') \psi_{\tau} ( Y_{it} - \bm{d}_{it}' \bm{\delta} - \bm{x}_{it}' \bm{\beta} - \bm{z}_{i}' \bm{\alpha}) = o_p(a_{NT}),   
\end{equation}
where $\psi_{\tau}(u) = \tau - I(u < 0)$ is the quantile influence function and $a_{NT} \to 0$ as $N$ and $T$ go jointly to infinity under the rates of convergence obtained in Kato, Galvao and Montes-Rojas (2012). It follows that $E(M_{NT}(\bm{\delta}_0,\bm{\beta}_0,\bm{\alpha}_0)) = 0$, and therefore, $M_{NT}(\bm{\delta}(\tau),\bm{\beta}(\tau),\bm{\alpha}(\tau))$ is an unbiased estimating function for the parameter of interest $(\bm{\delta}_0(\tau)',\bm{\beta}_0(\tau)',\bm{\alpha}_0(\tau)')$ provided that $T$ is sufficiently large.  

The approach is motivated by the fact that standard panel transformations are not available in quantile regression. Therefore, several papers in the literature estimate jointly $p=p_x+p_d$ slopes and $N$ individual effects. (The interested reader can find alternative approaches in Abrevaya and Dahl (2008)\nocite{Abre05}, Canay (2011)\nocite{iC11}, Chernozhukov, Fern\'{a}ndez-Val, Hahn, and Newey (2013)\nocite{vCwN13}, among others). Note that the model cannot include an overall intercept, because the intercept and the $N$-dimensional vector of parameters, $\bm{\alpha}(\tau)$, are not jointly identifiable or estimable. In large $N$ and small $T$ settings, it is expected that the previous approach create biases due to the estimation of incidental parameters. %However, Graham, Hahn, and Powell (2009)\nocite{bG09} show that there is no incidental parameter problem in a non-differentiable panel data model and Galvao, Lamarche, and Lima (2013) show that the performance of the fixed effects estimator improves rapidly under moderate $T$ in location-scale shift models. 

For the previous reason, we propose below an approach that improves the performance of the fixed effects estimator. Shrinkage of the individual effects towards zero can reduce estimation bias of the slope parameter when $T$ is small. In what follows, the vector of explanatory exogenous variables $\bm{x}_{it}$ includes a constant 1 and might consist on (i) time invariant covariates, $\bm{x}_i$, (ii) baseline characteristics and a deterministic function of time, $\bm{x}_i \cdotp t$, or (iii) time-varying covariates. Let $\bm{\vartheta} = (\bm{\delta}',\bm{\beta}')'$ and $\bm{V}_{it} = (\bm{d}_{it}',\bm{x}_{it}')'$ be a random vector taking values in $\mathcal{V} \subseteq \RR^p$. Moreover, let $\bm{X}_{it} = (\bm{V}_{it}',\bm{z}_i')'$ and $\bm{Z}_{i} = (\bm{0}',\bm{z}_{i}')'$ be a sparse vector of dimension $p+N$. The penalized panel quantile regression estimator (see, e.g., Harding and Lamarche (2017)\nocite{harding2016}, Lamarche (2010)) can be obtained as a solution of the following estimating equation:
\begin{equation}
M_{NT}(\bm{\theta},\lambda) = - \frac{1}{NT} \sum_{i=1}^N \sum_{t=1}^{T} \bm{X}_{it} \psi_{\tau} (Y_{it} - \bm{X}_{it}' \bm{\theta}) + \frac{\lambda}{N} \sum_{i=1}^N \bm{Z}_{i} \psi_{\tau}(\bm{z}_i' \bm{\alpha})  = o_p(a_{NT}) \label{negsubcon:pen}
\end{equation}
where $\lambda \in \RR_{+}$ is a penalty parameter, $\bm{\theta}(\tau) = (\bm{\vartheta}(\tau)',\bm{\alpha}(\tau)')'$ is contained in the parameter space $\bm{\Theta}$ and $a_{NT} \to 0$ as $N,T \to \infty$. In general, the solution of \eqref{negsubcon:pen}, $\hat{\bm{\theta}}(\tau)$, can depend on $\lambda$ but we assume the tuning parameter fixed and supress the dependence for notational convenience.  

\subsection{Attrition}
Suppose now that we have a random sample of individuals who are observed in the first occasion when $t=1$.  The probability of staying in the panel for unit $i$ at time $t$ is,
\begin{equation}
\pi_{0,it} = P(s_{it} = 1 | s_{it-1}=\hdots=s_{i2}=1, \bm{W}_{it}, \bar{\bm{V}}_{i}), \label{pi1}
\end{equation}
where $s_{it} = 1$ if and only if the response variable $Y_{it}$ is observed and 0 otherwise, $\bar{\bm{V}}_{i} = (\bm{V}_{i1}',\hdots,\bm{V}_{iT}')'$ is a vector of observed independent variables and $\bm{W}_{it}$ is a vector of variables that might include latent and observed responses depending on the assumptions associated with the missing data process. For instance, as explained in detail below, $\bm{W}_{it} = (Y_{it-1}, Y_{it-2}, \hdots)'$ in panel data models with selection on observables and $W_{it} = Y_{it}$ in models with selection on unobservables, because $Y_{it}$ is a latent variable for subjects $1 \leq i \leq N$ who dropped the panel at time $t > 1$.

Suppose there exists a monotone missing data pattern as in Robins, Rotnitzky and Zhao (1995)\nocite{Robins1995}. This refers to a situation where once a subject leaves the panel, the return into the sample is not possible. Suppose, for instance, that at time $t=1$, a random sample of $N$ subjects is drawn from the population. At $t=2$, a number of subjects drop out and they are not part of the panel at $t \in \{3,4, \hdots \}$. At $t=3$, other subjects drop out and are out of the sample at $t \in \{4,5,\hdots \}$, etc. Under a monotone missing data pattern, equation \eqref{pi1} can be written as, $\pi_{0,it} = P(s_{it} = 1 | s_{it-1}=1, \bm{W}_{it}, \bar{\bm{V}}_{i}) > 0$, where the strict inequality for all $t=1,...,T$ is required to guarantee the existence of a consistent estimator of the quantile treatment effect, $\bm{\delta}(\tau)$. 

\begin{assumption}\label{A1}
The probability $\pi_{0,it}$ is bounded away from 0, i.e. $\pi_{0,it} > \sigma > 0$ for $i=1,\hdots,N$ and $t=1,\hdots,T$. Moreover, $s_{it} = 0$ implies $s_{it+1} = 0$ for $t = 1,\hdots,T$.
\end{assumption}

Two models have been used for inference in panel data models. Identification results in the presence of missing data are obtained based on selection on observables (e.g., Fitzgerald, Gottschalk, and Moffitt 1998\nocite{fitzgerald1998}), which is also known as missing at random mechanism or simply MAR (Rubin 1976\nocite{Rubin76}, Robins, Rotnitzky and Zhao 1995\nocite{Robins1995}). It implies that $s \perp Y$ conditional on independent variables and observed response variables. The attrition probability can be written as $P(s_{it} = 1 | \bm{Y}_i, \bar{\bm{V}}_{i}) =  P(s_{it} = 1 | \bm{Y}_{i,t-1}, \bar{\bm{V}}_{i})$, where $\bm{Y}_{i} = (Y_{i1},\hdots,Y_{it-1},Y_{it},\hdots,Y_{iT})'$ and $\bm{Y}_{it-1} = (Y_{i1},\hdots,Y_{it-1})'$. The second model is introduced in Hausman and Wise (1979)\nocite{jH79} and it allows for the missing data process to be conditionally dependent of the missing responses. A simplified version of the model, for $T=2$ and $s_{i1}=1$ for all $i$, is:
\begin{eqnarray} 
Y_{it} & = & \bm{d}_{it}' \bm{\delta} + \bm{x}_{it}' \bm{\beta} + \alpha_i + u_{it}, \; \; \; t = \{1,2\} \label{HW:eq1} \\
s_{i2} & = & 1\{ \rho Y_{i2} + \bm{x}_{i2}' \bm{\gamma} + v_{it} > 0 \}. \label{HW:eq3}
\end{eqnarray}
It is immediately apparent that the error terms in equation \eqref{HW:eq1} at $t=2$ and equation \eqref{HW:eq3} are not independent, leading to selection issues. To see this, we replace equation \eqref{HW:eq1} for $t=2$ in equation \eqref{HW:eq3} and obtain a ``reduced form'' equation for the attrition process: $s_{i2} = 1\{ \bm{d}_{i2}' (\rho \bm{\delta}) + \bm{x}_{i2}' (\rho \bm{\beta} + \bm{\gamma}) + \rho \alpha_i + \rho u_{i2} + v_{it} > 0 \}$. In terms of equation \eqref{pi1} under Assumption \ref{A1}, $Y_{i2} = W_{i2}$ and $\bm{x}_{i2} = \bar{\bm{V}}_i$. To consistently estimate the parameters of the model and provide asymptotically efficient estimates, Hausman and Wise (1979) propose a maximum likelihood procedure for a random effects specification that allows testing for the presence of attrition. Fixed effects specifications might help in reducing biases but do not eliminate issues associated with attrition.

These two models rely on assumptions on the missing data process and, for consistent estimation, we do not require additional data as in Ridder (1992)\nocite{gRidder92}, Nevo (2003)\nocite{aNevo03} and Bhattacharya (2008)\nocite{Bhattacharya2008}. When additional data (e.g., ``refreshment" samples) are available, it is possible to correct panel data estimators to avoid biased and inconsistent results in models with both selection on observables and unobservables. Hirano et al. (2001)\nocite{kHirano01} state conditions under which the attrition function can be semi-parametrically identified in a model with selection on unobservables. 

\subsection{Identification}

It has been noted that in the MAR model, it is not possible to introduce dependence of the missing data process $s_{it}$ on $y_{it}$ because it is not observed for all the individuals. Also, the Hausman-Wise (HW) selection on unobservables model depends on parametric assumption and refreshment samples are not always available to practitioners. While the MAR and HW selection models have been extensively investigated and extended for classical conditional mean models, the relatively new literature on panel quantile models does not offer correction for potential inconsistencies arising from unobservables. 

Considering the ideal situation where the probability of dropping out the panel is known, we present an identification result for general patterns of missing data. Consider a slightly different equation \eqref{negsubcon:pen}:
\begin{equation}
M_{NT}(\bm{\theta}(\tau),\bm{\pi}_{0})) = - \frac{1}{NT} \sum_{i=1}^N \sum_{t=1}^{T} \left( \frac{s_{it}}{\pi_{0,it}} \bm{X}_{it} \psi_{\tau} (Y_{it} - \bm{X}_{it}' \bm{\theta}) - \lambda \bm{Z}_{i} \psi_{\tau}(\bm{z}_i' \bm{\alpha})  \right)
\end{equation}
and let $E( M_{it}(\bm{\theta}(\tau),\bm{\pi}_{0})) := E \left( \frac{s_{it}}{\pi_{0,it}} \bm{X}_{it} \psi_{\tau} (Y_{it} - \bm{X}_{it}' \bm{\theta}) - \lambda \bm{Z}_{i} \psi_{\tau}(\bm{z}_i' \bm{\alpha}) \right)$. 

The result of this section requires the following additional conditions:

\begin{assumption}\label{A2}
The probability $\pi_{0,it} = g(\kappa(\bm{W}_{it})' \bm{\gamma}))$ where $g: \RR \to \RR$ is a known, differentiable, strictly increasing function such that $\lim_{c \to -\infty} g(c) = 0$ and $\lim_{c \to +\infty} g(c) = 1$. 
\end{assumption}

\begin{assumption}\label{A0}
For $t-1 < t$, there is an independent sample $\{ \bm{W}_{ih_i} \}_{i=1}^{N}$ from the same population than $\{ \bm{W}_{it} \}_{i=1}^N$, where $h_{i} = \sup \{ h_{ij} :  | h_{ij} - t | < \epsilon \}$ for a collection of dates $\{h_{ij}\}_{j=1}^{J_i}$ between $t-1$ and $t$. It follows that $P(s_{it} = 1 | \bm{W}_{it}, \bar{\bm{V}}_i) - P(s_{it} = 1 | \bm{W}_{ih_i}, \bar{\bm{V}}_i) = 0$ almost surely.
\end{assumption}

\begin{assumption}\label{A3}
Let $\bm{\theta} = (\bm{\delta}',\bm{\beta}',\bm{\alpha}')'$, where $\bm{\alpha} = (\alpha_1,\hdots,\alpha_N)'$ and $\bm{\theta} \in \mathcal{A}^N \times \mathcal{B} \times \mathcal{D}$, where $\mathcal{A}$ is a compact subset of $\RR$, $\mathcal{A}^N$ is a product of $N$ copies of $\mathcal{A}$, and $\mathcal{B}$ and $\mathcal{D}$ are compacts subsets of $\RR^{p_d}$ and $\RR^{p_x}$. Then, $\bm{\theta}_0$ uniquely solves $E\left( \bm{X}_{it} \psi_{\tau} ( Y_{it} - \bm{X}_{it}' \bm{\theta} ) + \lambda \bm{Z}_{i} \psi_{\tau} (\bm{z}_i' \bm{\alpha}) \right) = 0$.
\end{assumption}

Assumption \ref{A2} is similar to condition A3 in Nevo (2003) and it includes several selection models including parametric functions as the logistic model used later. If we let $\bm{\gamma} \in \bm{\Gamma} \subset \RR^M$ and $\bm{\gamma}_0$ be a maximizer of $E\( s_{it} \log(\pi_{0,it}) + (1 - s_{it}) \log(1 - \pi_{0,it}) \)$, we have a condition similar to Assumption 3.2 in Wooldridge (2007)\nocite{Wooldridge20071281}. Condition \ref{A0} requires the availability of measures of the dependent variables over small time intervals and it implies a ``local" time-homogeneity condition. More specifically, it can imply that the joint distribution of $Y_{it},s_{it} | \bar{\bm{V}}_i$ is identical to the joint distribution of $Y_{ih_i},s_{it} | \bar{\bm{V}}_i$. Assumption \ref{A3} implies that the quantile regression model is identified under no missing data. Note that $\lambda=0$ gives the standard condition $E(\bm{X}_{it} \psi_{\tau} (u_{it}(\tau))) = 0$, where $u_{it}(\tau) := Y_{it} - \bm{X}_{it}'\bm{\theta}(\tau)$. When $\lambda > 0$, we require that the $\alpha_i$'s are conditionally independent of $\bm{X}_{it}$ for point identification of the slope parameters. The condition $E(\lambda \bm{Z}_i (\psi_{\tau}(\alpha_{i}))) = 0$ because it is assumed that $E(I(\alpha_i \leq 0)) = \tau$, and it implies that the $\tau$-th conditional quantile of $\alpha_i$ is equal to zero. In the case that $\alpha_i(\tau) = \alpha_{i0}$ for all $\tau$, as in Koenker (2004) and Lamarche (2010), the individual effects are assumed to be drawn from a zero-median distribution function independent of $\bm{X}_{it}$. It is worth noting that the previous assumption can be replaced by a sparsity condition on the parameters of the model, with $\alpha_{i0} = 0$ for all $1 \leq i \leq N$.

\begin{proposition}\label{P1}
Under Assumptions \ref{A1}-\ref{A3}, the treatment effect parameter of the quantile regression model in model \eqref{eq:qrmodel}, $\bm{\delta}(\tau)$, is identified using the sample $\{ \bm{W}_{ih_i} \}$.
\end{proposition}

The result in Proposition \ref{P1} leads to a two-step estimator which extends existing results to the case of selection on unobservables. This is possible under Assumption \ref{A0}, which it is argued to be satisfied in our application by the availability of a `streaming sample' as explained in Section 4.4.
 
\subsection{A Quantile Estimator}

Similarly to Lipsitz et al. (1997)\nocite{sLipsitz1997} and Maitra and Vahid (2006), our method adopts Robins, Rotnitzky and Zhao (1995)\nocite{Robins1995} idea to weight uncensored observations by the inverse probabilities. In contrast with existing work, attrition can depend on variables that are not observed when the subjects drop. 

As before, we first assume that the probability of dropping out of the panel is known. In this case, the quantile regression coefficient $\bm{\theta}_0(\tau) = (\bm{\vartheta}_0(\tau)', \bm{\alpha}_0(\tau)')'$ can be estimated by minimizing the following objective function:
\begin{equation}
Q_{NT}(\bm{\theta}(\tau),\bm{\pi}_{0}) = \frac{1}{NT} \sum_{i=1}^N \sum_{t=1}^{T} \left( \frac{s_{it}}{\pi_{0,it}} \rho_{\tau} ( Y_{it} - \bm{X}_{it}' \bm{\theta}) + \lambda \rho_{\tau}(\bm{z}_i' \bm{\alpha})  \right),  \label{eq:qrweights1} 
\end{equation}
where $\bm{X}_{it} = (\bm{V}_{it}',\bm{z}_{i}')'$ and $\bm{V}_{it}$ contains an intercept. The solution is $\hat{\bm{\theta}}(\tau,\lambda)$. We concentrate our attention to $\{ \hat{\bm{\theta}}(\tau,\lambda), \; \lambda \in [\lambda_L, \infty) \}$, where $\lambda_L \in (0,\infty)$ is a deterministic constant subject to identifiability restrictions. It should be noted that, because the model contains an intercept, we do not attempt to estimate $\hat{\bm{\theta}}(\tau,\lambda)$ for all $\lambda \in [0,\infty)$. Our estimator is defined for $\lambda > 0$, although $\lambda$ can be very small. When $\lambda_L > 0$, there are $\hat{\alpha}_i$'s that are exactly zero, which is equivalent to a model with $m < N$ individual effects. This allows identification of the intercept and the $N$-dimensional vector of parameters $\bm{\alpha}(\tau)$.

Alternatively, the objective function \eqref{eq:qrweights1} can be written simply as,
\begin{equation}
Q_{NT}(\bm{\theta}(\tau),\bm{\pi}_{0}) = \frac{1}{NT} \sum_{i=1}^N \sum_{t=1}^{T_i} \left( \rho_{\tau} ( \tilde{Y}_{it} - \tilde{\bm{X}}_{it}' \bm{\theta}) + \lambda \rho_{\tau}(\bm{z}_i' \bm{\alpha}) \right),  \label{eq:qrweights2} 
\end{equation}
where $\tilde{Y}_{it} = s_{it} Y_{it} / \pi_{0,it}$ and $\tilde{\bm{X}}_{it} = s_{it} \bm{X}_{it} /\pi_{0,it}$. The estimating equation can be expressed as,
\begin{equation}
M_{NT}(\bm{\theta}(\tau),\bm{\pi}_{0}) = - \frac{1}{NT} \sum_{i=1}^N \sum_{t=1}^{T_i} \left( \tilde{\bm{X}}_{it} \psi_{\tau} (\tilde{Y}_{it} - \tilde{\bm{X}}_{it}' \bm{\theta}) - \lambda \bm{Z}_{i} \psi_{\tau}(\bm{z}_i' \bm{\alpha}) \right)
\end{equation}
where the vector $\bm{Z}_{i} = (\bm{0}',\bm{z}_{i}')'$ is defined as before. Naturally, $M_{NT}(\bm{\theta}(\tau),\bm{\pi}_{0})$ might not be equal to zero, so we minimize instead $Q_{NT}(\bm{\theta}(\tau),\bm{\pi}_{0})$. Because the objective function is defined in terms of variables which are reweighted by the inverse probability of staying in the sample, existing linear programming algorithms for panel quantiles can be employed including the functions in the {\tt R} package {\tt quantreg} (Koenker 2013\nocite{Koenker16}). The penalty form is chosen to preserve the linear programming problem, and therefore it has computational advantages. Note that if $\tau=1/2$, we obtain a lasso-type penalty whose statistical advantages are well documented in the literature (see Koenker (2004), Belloni and Chernozhukov (2011), among others).

In reality, the propensity score $\pi_{0,it}$ is unknown and needs to be estimated. There are several alternatives available for estimating $\pi_{0,it}$ based on the assumed missing data mechanisms (Robins, Rotnitzky and Zhao (1995), Nevo (2002), Deng et al. (2013), among others). We propose a two-step estimator obtained as follows:

\begin{itemize}
\item[Step 1:] Estimate $\pi_{0,it}$ by either parametric or nonparametric methods considering $\{ (s_{it},\bm{W}_{ih_i},\bar{\bm{V}}_i) \}$ under different assumptions on the attrition process. This step can accommodate MAR and HW models. We denote the estimate probability by $\hat{\pi}_{it}$.

\item[Step 2:] Let $\lambda_L > 0$. For $\lambda \in [\lambda_L, \infty)$, estimate $\bm{\theta}_0(\tau)$ by finding the argument that minimizes 
\begin{equation}
\arg \min_{\bm{\theta} \in \bm{\Theta}} \sum_{i=1}^N \sum_{t=1}^{T} \left( \frac{s_{it}}{\hat{\pi}_{it}} \rho_{\tau} ( Y_{it} - \bm{X}_{it}' \bm{\theta}) + \lambda \rho_{\tau}(\bm{z}_i' \bm{\alpha}) \right).
\end{equation}
The solution is defined as the weighted penalized quantile regression estimator (WPQR) for an unbalanced panel data model:
\begin{equation}
\hat{\bm{\theta}}(\tau,\hat{\bm{\pi}}) = (\hat{\bm{\vartheta}}(\tau,\hat{\bm{\pi}})', \hat{\bm{\alpha}}(\tau,\hat{\bm{\pi}})')'. \label{wpqr}
\end{equation}
\end{itemize}

In Step 1, it is possible to estimate the probabilities $\pi_{0,it}$ based on an additive non-ignorable model, which contains the MAR mechanism and the model of Hausman and Wise (1979) as special cases. Consider, again for simplicity, a two period panel data and let $\pi_{0,it} = g(Y_{it},\bar{\bm{V}}_i; \gamma)$ where $g(\cdot)$ is a known link function. At $t=1$, $E ( s_{i1}/\pi_{0,i1} - 1 | W_{ih_i}, \bar{\bm{V}}_i) = 0$ is identifiable from the unbalanced data because $W_{ih_i} = W_{i1} = Y_{i1}$ for all $1 \leq i \leq N$. At time $t=2$, identification requires a ``refreshment" sample from population, because $Y_{i2}$ is not observed for some $1 \leq i \leq N$. Under the existing assumptions, an identifying moment is $E (s_{i2} / \pi_{0,i2} - 1 | W_{ih_i}, \bar{\bm{V}}_i) = 0$, where $W_{ih_i} = W_{i2} = Y_{i2}$ if $s_{i2} = 1$ and $W_{ih_i} = Y_{ih_i}$ if $s_{i2}=0$. 

Under the assumption that the propensity score follows a parametric model under Assumption \ref{A0}, $P(s_{it} = 1 | \bm{W}_{it}) = p(\bm{W}_{it}' \bm{\gamma})$ where $s_{it} = 1$ if the data is not missing, $p(\cdot)$ is a known link function and $\bm{\gamma}$ is a vector of unknown parameters. It is straightforward to augment the model with desired transforms of $\bm{V}_{it}$, denoted by $\dot{\bm{V}}_{it}$, and then form $\dot{\bm{W}}_{it}$. For instance, the transforms of the covariates could be equal to a vector of independent variables that includes $\bm{x}_{it}$ and $\bm{x}_{it}^2$ as in Chernozhukov and Hong (2002)\nocite{ChernozhukovHong02}. The parametric estimation of the propensity score can be done using the Manski maximum score method for a model with differences. In cases where the propensity score is unknown, we propose to estimate it using nonparametric or semiparametric methods. Although a root-$n$ consistent estimator can be obtained by Maximum Likelihood, it is possible to prove the 4th-root uniform consistency of a non-parametric estimator for $\pi_0$ as in Galvao, Lamarche and Lima (2013). It can be obtained by applying non-parametric methods (e.g., Kernel or Spline regression or Generalized Additive Models) to data on $s_{it}$ and $\bm{W}_{ih_i}$. 

The procedure can be simply modified to estimate a model with individual location shifts. We estimate the probability of attrition using the method described in Step 1, and then, in Step 2, we estimate $\bm{\theta}_0(\tau)$ as follows: 
\begin{itemize}
\item[Step 2']: For $\lambda \in [\lambda_L,\infty)$, estimate $\bm{\theta}_0(\tau)$ by finding the argument that minimizes
\begin{equation*}
\arg \min_{\bm{\theta} \in \bm{\Theta}} \sum_{j=1}^J \sum_{i=1}^N \sum_{t=1}^T \frac{s_{it}}{\hat{\pi}_{it}} \omega_j \rho_{\tau_j} ( Y_{it} - \bm{X}_{it}' \bm{\theta}) + \lambda \sum_{i=1}^N | \bm{z}_i'  \bm{\alpha} |.
\end{equation*}
where $\omega_j$ is a weight given to the $j$-th quantile $\tau_j \in (0,1)$ and $J$ is the number of quantiles $\{\tau_1,\tau_2,\hdots,\tau_J\}$ simultaneously estimated. 
\end{itemize}

The choice of the weights $\bm{\omega} = (\omega_1,\omega_2,\hdots,\omega_J)'$ is somewhat analogous to the choice of discretely weighted L-statistics (Koenker 2004). At the cost of losing efficiency, a practical alternative is to weight equally all quantiles by setting $\omega_j = J^{-1}$ for all $1 \leq j \leq J$.

\subsection{Asymptotic Theory}\label{section:asymptotics}

We consider the following regularity conditions for the consistency of the proposed estimator. Throughout this section, $\| \cdotp \|_1$ stands for the $\ell_1$-norm.

\begin{assumption}\label{A4}
$\{ (\bm{V}_{it}', Y_{it}) \}$ are independent across individuals and independently and identically distributed (i.i.d.) within each individual.
\end{assumption}
\begin{assumption}\label{A5}
There exists a constant $M$ such that $\max \| \bm{V}_{it} \| < M$, where $\bm{V}_{it} = (\bm{d}_{it}',\bm{x}_{it}')'$.
\end{assumption}
\begin{assumption}\label{A6}
Let $\omega_{it}(\bm{\gamma}) := s_{it} / \pi_{0,it}(\bm{\gamma})$. For each $\eta > 0$, 
\begin{equation*}
\epsilon_\eta := \inf_{i\geq1} \inf_{\| \bm{\theta} \|_1 = \eta} E \left[ \int_0^{\bm{X}_{i1}' \bm{\theta}} \omega_{it}(\bm{\gamma}) \left( F_i(s | \bm{X}_{i1}) - \tau \right) ds + \lambda \int_0^{\bm{z}_i' \bm{\alpha}} \left( G_i(s | \bm{V}_{i1}) - \tau \right) ds \right] > 0,
\end{equation*}
where $F_i$ is defined as a conditional distribution of $u_{it}$ and $G_i$ as the conditional distribution of $\alpha_i$. The distribution of $\alpha_i$ has a zero quantile function conditional on $\bm{V}_{it}$. The conditional densities $f_i$ and $g_i$ are continuous, uniformly bounded away from 0 and $\infty$, with continuous derivatives everywhere.
\end{assumption}

Assumption \ref{A4} is standard and has been used in Fernandez-Val (2005)\nocite{Ivan2005}, Hahn and Newey (2004)\nocite{Newey2004}, Kato, Galvao and Montes-Rojas (2012) and Galvao, Lamarche and Lima (2013). As in Galvao et al., we consider the case of no temporal dependence, and thus, we focus our attention on attrition in static panel quantile models. It is possible to allow dependence across time by applying stochastic inequalities for $\beta$-mixing sequences, as in Theorem 5.1 in Kato, Galvao and Montes-Rojas (2012). We shall stress that the restriction that $T$ grows at most polynomially in $N$ does not change and the dependence case leads, as expected, to a different asymptotic covariance matrix than the one obtained under Assumption \ref{A4}. Assumption \ref{A5} is also common in the literature (see, e.g., Koenker 2004, Lamarche 2010) and is important for the finite dimensional uniform convergence of the objective function. This assumption can be relaxed using a moment condition as in Fernandez-Val (2005) and Kato, Galvao and Montes-Rojas (2012). Condition \ref{A6} is an identification condition and is similar to Assumption (A3) in Kato, Galvao and Montes-Rojas (2012) and Condition 3 in Hahn and Newey (2004) when $\lambda \to 0$. The second term leads to point identification and it is similar to Condition A2 in Lamarche (2010). 

The following result states the consistency of the estimator:

\begin{theorem}\label{T1}
Under Assumptions \ref{A1}, \ref{A2}, \ref{A0}, \ref{A4}, \ref{A5}, and \ref{A6}, as $N$ and $T$ goes jointly to infinity with $\log(N)/T \to 0$, the weighted penalized quantile regression estimator (WPQR) for an unbalanced panel data model, $\hat{\bm{\vartheta}}(\tau,\hat{\bm{\pi}})$, is consistent. 
\end{theorem}

The result shows that the weighted quantile regression estimator is consistent. The result is shown using the arguments in Theorem 3.1 in Kato, Galvao and Montes-Rojas (2012) and Theorem 1 in Galvao, Lamarche and Lima (2013). The restriction on the growth of $T$, which should be denoted by $T_N$ because it depends on the number of subjects, is similar to the literature. Improvements based on $\lambda$ selection is out of the scope of this paper.

For the convergence in distribution of the proposed estimator, consider the following additional conditions.

\begin{assumption}\label{A8} There exists positive definite matrices $\bm{D}_0$ and $\bm{D}_1$ such that:
\begin{eqnarray*}
\bm{D}_{0} & = & \lim_{N \to \infty} \frac{1}{N} \sum_{i=1}^N E \left\{  \left( \tilde{\bm{V}}_{it} - \tilde{\bm{E}}_i \varphi_i^{-1} \right) \left( \tilde{\bm{V}}_{it} - \tilde{\bm{E}}_i \varphi_i^{-1} \right)' - \left( \tilde{\bm{E}}_i (\lambda/\varphi_i) \right) \left(\tilde{\bm{E}}_i (\lambda/\varphi_i) \right)' \right\}, \\
\bm{D}_{1} & = & \lim_{N \to \infty} \frac{1}{N} \sum_{i=1}^N  \left( \tilde{\bm{J}}_i - \tilde{\bm{E}}_i \varphi_i^{-1} \tilde{\bm{E}}_i' \right),  
\end{eqnarray*}
where $\tilde{\bm{V}}_{it} = [s_{it} / \pi_{0,it}(\bm{\gamma})] \bm{V}_{it}$, $\tilde{\bm{E}}_{i} = E( f_i(0 | \bm{X}_{it}) \tilde{\bm{V}}_{it}$), $e_i = E( [s_{it} / \pi_{0,it}(\bm{\gamma})] f_i(0 | \bm{V}_{it}))$, $g_i = E( g_i(0 | \bm{X}_{it}) )= E( g_i(0) )$, $\varphi_i := e_i - \lambda g_i / \sqrt{T}$, and $\tilde{\bm{J}}_i = E( f_i(0 | \bm{X}_{it}) \tilde{\bm{V}}_{it} \tilde{\bm{V}}_{it}')$. 
\end{assumption}
\begin{assumption}\label{A7}
Let $\lambda_T$ be a given tuning parameter for a panel data model with $T$ observations for each subject. Then, the regularization parameter $\lambda_{T} / \sqrt{T} \to \lambda > 0$.  
\end{assumption}

Assumption \ref{A8} is standard in the quantile regression literature and it implies that the limiting matrices exists and are non-singular. It is also implicitly assumed that the minimum eigenvalue of $\bm{D}_{1,N}$ is bounded away from zero uniformly over $N \geq 1$.  The matrices are similar to the ones in Condition (B3) in Kato, Galvao and Montes-Rojas (2012) and Condition B6 in Galvao, Lamarche and Lima (2013) when $s_{it} = 1$ for all $t > 1$ and $\lambda \to 0$. Lastly, Assumption \ref{A7} is a condition used for penalized estimators and it has been previously assumed in Knight and Fu (2000)\nocite{knight2000} and Koenker (2004) to achieve square root-$n$ consistency for the penalized estimator. As shown in the proof of Theorem \ref{T2}, the rate of growth of $\lambda_T$ determines a limiting distribution of the penalized estimator that is different than the fixed effects quantile regression estimator. For asymptotic normality, we require $\lambda_T = O(\sqrt{T})$, although for consistency, $\lambda_T$ can grow faster.

The following result obtains the asymptotic distribution of the proposed estimator:

\begin{theorem}\label{T2}
Under the conditions of Theorem \ref{T1} and Assumptions \ref{A8} and \ref{A7}, provided that $N^2 (\log(N))^3/T \to 0$ as $N$ and $T$ go jointly to infinity, the weighted penalized quantile regression estimator (WPQR) for an unbalanced panel data model, $\hat{\bm{\vartheta}}(\tau,\hat{\bm{\pi}})$, converges in distribution to a Gaussian random vector with mean $\bm{\vartheta}(\tau)$ and covariance matrix $\tau (1 - \tau) \bm{D}_{1}^{-1} \bm{D}_{0} \bm{D}_{1}^{-1}$.
\end{theorem}

The components of the asymptotic covariance matrices in Theorem \ref{T2} can be estimated using standard methods (Koenker 2005, \S3), and therefore, they will not be discussed in this article. The proof of Theorem \ref{T2} is based on a parametric first stage as in Nevo (2003) but the result can be extended to estimating the propensity score by non-parametric methods as in Tang et. al. (2012)\nocite{jW12}. In this case, we need assumptions on the smoothness of the propensity score function and conditions on bounded support and derivatives. Also, it requires that $\sup_i \| \hat{\pi}_{it} - \pi_{i0} \|_{\infty} = o_p(T^{-1/4})$, where $\| \hat{\pi}_{it} - \pi_{i0} \|_{\infty} = \sup_{z \in \mathcal{Z}} | \hat{\pi}(z) - \pi_0(z) |$ for a generic vector $z$ and a given function $\pi(\cdot)$.

\section{Simulation Studies}
This section reports the results of several simulation experiments designed to evaluate the performance of the method in finite samples. First, we investigate the small sample performance of the penalized estimator relative to the existing fixed effects estimator in cases with and without missing data. Second, we briefly investigate the bias and root mean square error (RMSE) of the estimator in models with endogenous individual effects. We are especially interested in comparing the performance of the method with respect to existing quantile regression estimators. Finally, we will contrast the performance of the quantile regression estimator in the case of selection on unobservables using refreshment samples. 

We focus on the simulation experiments that can lead to close comparisons of results with the one obtained by Kato, Galvao, and Montes-Rojas (2012) and Kyriazidou (1997) in the presence of missing data. We generate the dependent variable as:
\begin{eqnarray}
y_{it} & = & s_{it}  (\alpha_i + \beta_0 + \beta_1 x_{it} + (1 + \gamma x_{it}) u_{it}), \label{kgm1} \\
s_{it} & = & 1\{ \rho_0 y_{it}^\ast + \rho_1 y_{it-1} + \theta_1 x_{it} + \theta_2 \alpha_i - v_{it} > 0 \}  \label{kgm2}  \\
x_{it} & = & \pi \alpha_i + z_{it}, \label{kgm3}
\end{eqnarray}
where $s_{i1} = 1$ for all $i$, $z_{it} \sim \chi_3^2$, and $\alpha_i \sim \mathcal{U}[0,1]$. The distribution of the error term $u_{it}$ is i.i.d. $\chi_3^2$ or Cauchy. Then, the distribution of $u_{it}$ is changed in the simulation designs following closely Kato, Galvao, and Montes-Rojas (2012), although we do not consider the Normal case because the bias of the fixed effects quantile regression estimator is negligible in the simulations. It is assumed that $\beta_0 = 0$, $\beta_1 = 1$, $\gamma = 0.5$ and $\pi = 0.3$ to obtain the data generating process considered in Kato, Galvao, and Montes-Rojas (2012). In models with missing data, as in Kyriazidou (1997), the error term $v_{it}$ is distributed as logistic and the parameter of interest is $\beta_1$ in equation \eqref{kgm1}. The number of Monte Carlo experiments is 1000. 

\begin{figure}
\begin{center}
\centerline{\includegraphics[width=.9\textwidth]{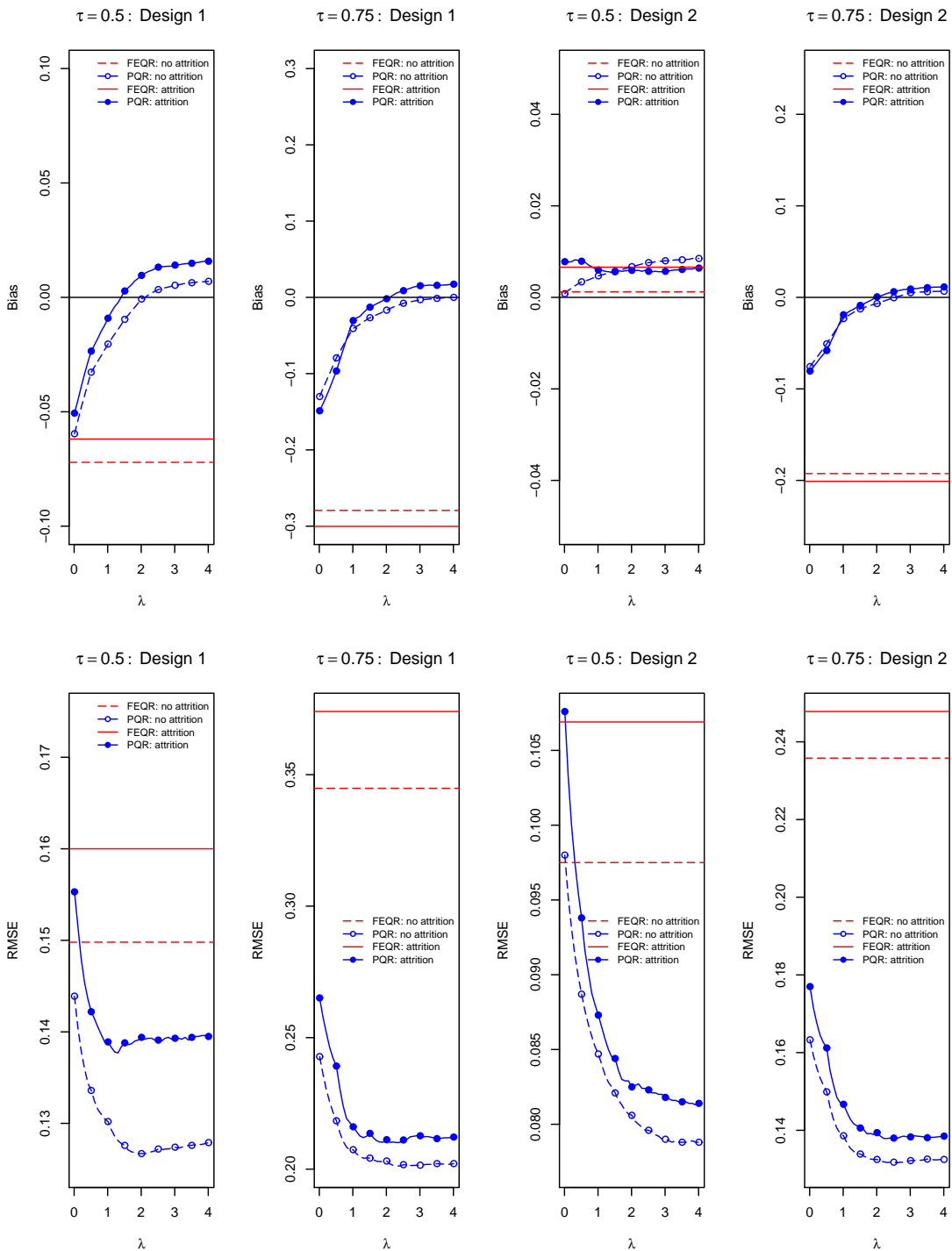}}
\caption{\emph{Small sample performance of the fixed effects estimator (FEQR) and penalized estimator (PQR) under Designs 1 and 2.}  \label{mc-fig1}}
\end{center}
\end{figure}

\subsection{Case 1: Shrinkage and missing data} We begin by emphasizing the difference between the fixed effects estimator and the penalized estimator for a model with individual effects. The fixed effects estimator was proposed in Koenker (2004) and further analyzed in Kato, Galvao, and Montes-Rojas (2012). The penalized quantile estimator is also proposed in Koenker (2004) and further investigated in Lamarche (2010). These estimators are different, and their differences lead to different small and large sample performances. The penalized estimator was introduced as a way of reducing the influence of the nuisance parameters in models with large $N$ and small $T$. Consider the following variations of the model \eqref{kgm1}-\eqref{kgm3} for $N=200$ and $T=5$:

\begin{description}
\item[Design 1.a (no attrition)] We focus on the case of no missing data assuming that $s_{it}=1$ for all $(i,t)$. The distribution of the error term $u_{it}$ is assumed to be $\chi_3^2$ (Table 4 in Kato, Galvao, and Montes-Rojas 2012).     
\item[Design 1.b (attrition)] The design is similar to Design 1.a but we generate missing data following equation \eqref{kgm2}. We assume $\theta_1 = \theta_2 = 1$ and $\rho_0 = \rho_1 = 0$, generating an average proportion of missing data of 15.4\%.
\item[Design 2.a (no attrition)] We focus on the case of no missing data assuming that $s_{it}=1$ for all $(i,t)$. We assume that the distribution of the error term $u_{it}$ is Cauchy as in Table 4 in Kato, Galvao, and Montes-Rojas (2012). 
\item[Design 2.b (attrition)] The design is similar to Design 2.a but, as in Design 1.b, we generate missing data following equation \eqref{kgm2}. We assume $\theta_1 = \theta_2 = 1$ and $\rho_0 = \rho_1 = 0$, generating an average proportion of missing data of 15.6\%. 
\end{description}

In this section, we compare the performance of the following estimators: (1) the fixed effects estimator defined in equation \eqref{eq:objectivef} (FEQR) and (2) the penalized quantile regression estimator (PQR) defined in \eqref{wpqr} but assuming that $\hat{\pi}_{it} = 1$ for all $(i,t)$. We report evidence on the performance for the penalized estimator for a series of tuning parameter values, $\lambda \in (0,4]$. Figure \ref{mc-fig1} presents the small sample performance of the fixed effects estimator and penalized estimator under Designs 1 and 2, with and without attrition.

In models without attrition, Figure \ref{mc-fig1} shows evidence quite consistent with Kato, Galvao, and Montes-Rojas (2012). The panels show that the fixed effects quantile regression (FEQR) estimator suffers from substantial bias. The extent of the bias varies with the quantiles at which the estimator is applied and the extent is determined by the specific distributional assumption of the error term. Note in particular that the bias of the fixed effects estimator can be as large as 28\% ($\tau=0.75$). 

In contrast, the penalized quantile regression estimator (PQR) corresponding to $\lim_{\lambda \to 0} \hat{\bm{\beta}}(\lambda)$ reduces the bias for increasing values of $\lambda$ (within the range considered in the simulations). Note in particular that even small increases of the penalty lead to very substantial improvements in both the bias profile and the RMSE. This further emphasizes that there are important distinctions between the FEQR and the PQR estimators. 

The significant gains from using PQR are present across Design 1.a and Design 2.a. It is particularly noteworthy that for some quantiles under some specifications the bias is nearly zero for the FEQR estimator, while at different quantiles the bias can be quite large either positive or negative. At the same time it is true that the PQR estimator has near zero bias in the cases where the FEQR estimator works well too, while almost completely removing the bias in the cases where the FEQR performs poorly. In all scenarios presented in our figures the RMSE decreases sharply and it is the case that a value of the tuning parameter exists such that the PQR dominates FEQR from both a mean and a RMSE perspective. 

The results for the case of missing data are also described in Figure \ref{mc-fig1}. We expect a slightly larger bias and higher RMSE than in Designs 1.a and 2.a because the models are estimated without the inverse propensity score weighting. This is exactly what we observe. For instance, the bias of FEQR now reaches 30\% and the PQR estimator exhibits small positive biases for large values of $\lambda$ at the 0.75 quantile. In terms of RMSE, we see an increase relative to the case of no attrition across different values of the tuning parameter $\lambda$. However, we continue to see that that even a small penalty leads to very substantial improvements in both the bias profile and the RMSE. 

\begin{table}
\begin{center}
\begin{tabular}{ cccccccccc } \hline
N & T & $\tau$ &  & \multicolumn{6}{c}{Quantile Regression Panel Data Methods}  \\
        &               &               &  &    QR      &       WQR     &       FE      &       WFE     &       PQR     &       WPQR        \\ \hline
200     &       5       &       0.50    &       Bias    &       0.172   &       0.194   &       -0.076  &       -0.072  &       0.115   &       0.045   \\
        &               &               &       RMSE    &       0.229   &       0.244   &       0.199   &       0.190   &       0.191   &       0.162   \\
200     &       25      &       0.50    &       Bias    &       0.160   &       0.172   &       -0.022  &       -0.022  &       0.074   &       0.009   \\
        &               &               &       RMSE    &       0.184   &       0.195   &       0.104   &       0.102   &       0.118   &       0.097   \\
500     &       5       &       0.50    &       Bias    &       0.178   &       0.196   &       -0.070  &       -0.072  &       0.124   &       0.049   \\
        &               &               &       RMSE    &       0.197   &       0.213   &       0.127   &       0.124   &       0.150   &       0.102   \\
500     &       25      &       0.50    &       Bias    &       0.151   &       0.167   &       -0.032  &       -0.031  &       0.066   &       0.001   \\
        &               &               &       RMSE    &       0.162   &       0.176   &       0.071   &       0.069   &       0.088   &       0.061   \\ \hline
200     &       5       &       0.75    &       Bias    &       0.165   &       0.191   &       -0.363  &       -0.215  &       0.088   &       0.022   \\
        &               &               &       RMSE    &       0.275   &       0.283   &       0.437   &       0.322   &       0.234   &       0.217   \\
200     &       25      &       0.75    &       Bias    &       0.161   &       0.172   &       -0.123  &       -0.055  &       0.076   &       0.006   \\
        &               &               &       RMSE    &       0.217   &       0.223   &       0.199   &       0.163   &       0.165   &       0.148   \\
500     &       5       &       0.75    &       Bias    &       0.183   &       0.200   &       -0.335  &       -0.202  &       0.098   &       0.031   \\
        &               &               &       RMSE    &       0.228   &       0.241   &       0.371   &       0.254   &       0.168   &       0.142   \\
500     &       25      &       0.75    &       Bias    &       0.148   &       0.165   &       -0.130  &       -0.055  &       0.062   &       0.003   \\
        &               &               &       RMSE    &       0.177   &       0.191   &       0.165   &       0.117   &       0.113   &       0.101   \\ \hline
200     &       5       &       0.90    &       Bias    &       0.146   &       0.182   &       -1.243  &       -0.729  &       0.029   &       0.020   \\
        &               &               &       RMSE    &       0.390   &       0.375   &       1.275   &       0.796   &       0.358   &       0.308   \\
200     &       25      &       0.90    &       Bias    &       0.145   &       0.155   &       -0.378  &       -0.126  &       0.055   &       0.007   \\
        &               &               &       RMSE    &       0.275   &       0.277   &       0.460   &       0.282   &       0.240   &       0.235   \\
500     &       5       &       0.90    &       Bias    &       0.183   &       0.201   &       -1.216  &       -0.724  &       0.063   &       0.040   \\
        &               &               &       RMSE    &       0.290   &       0.296   &       1.230   &       0.752   &       0.224   &       0.198   \\
500     &       25      &       0.90    &       Bias    &       0.142   &       0.159   &       -0.384  &       -0.122  &       0.054   &       0.016   \\
        &               &               &       RMSE    &       0.211   &       0.223   &       0.419   &       0.206   &       0.158   &       0.159   \\ \hline
\end{tabular}
\vspace{3mm}
\caption{\emph{Small sample performance of panel quantile methods in Design 3.} \label{mcI}} 
\end{center}
\end{table}

\begin{table}
\begin{center}
\begin{tabular}{ cccccccccc } \hline
N & T & $\tau$ &  & \multicolumn{6}{c}{Quantile Regression Panel Data Methods}  \\
        &               &               &  &    QR      &       WQR     &       FE      &       WFE     &       PQR     &       WPQR        \\ \hline
200     &       5       &       0.50    &       Bias    &       0.219   &       0.220   &       -0.056  &       -0.056  &       0.160   &       0.110   \\
        &               &               &       RMSE    &       0.256   &       0.256   &       0.169   &       0.172   &       0.209   &       0.183   \\
200     &       25      &       0.50    &       Bias    &       0.170   &       0.172   &       -0.021  &       -0.022  &       0.077   &       0.044   \\
        &               &               &       RMSE    &       0.187   &       0.188   &       0.090   &       0.090   &       0.110   &       0.095   \\
500     &       5       &       0.50    &       Bias    &       0.219   &       0.220   &       -0.054  &       -0.053  &       0.160   &       0.111   \\
        &               &               &       RMSE    &       0.233   &       0.234   &       0.115   &       0.115   &       0.179   &       0.142   \\
500     &       25      &       0.50    &       Bias    &       0.171   &       0.173   &       -0.022  &       -0.022  &       0.078   &       0.045   \\
        &               &               &       RMSE    &       0.179   &       0.180   &       0.060   &       0.060   &       0.094   &       0.071   \\  \hline
200     &       5       &       0.75    &       Bias    &       0.209   &       0.209   &       -0.315  &       -0.334  &       0.123   &       0.012   \\
        &               &               &       RMSE    &       0.290   &       0.291   &       0.401   &       0.416   &       0.234   &       0.221   \\
200     &       25      &       0.75    &       Bias    &       0.162   &       0.164   &       -0.103  &       -0.109  &       0.067   &       0.013   \\
        &               &               &       RMSE    &       0.203   &       0.204   &       0.170   &       0.174   &       0.142   &       0.130   \\
500     &       5       &       0.75    &       Bias    &       0.210   &       0.212   &       -0.314  &       -0.328  &       0.128   &       0.026   \\
        &               &               &       RMSE    &       0.244   &       0.245   &       0.348   &       0.361   &       0.177   &       0.133   \\
500     &       25      &       0.75    &       Bias    &       0.169   &       0.170   &       -0.100  &       -0.106  &       0.071   &       0.018   \\
        &               &               &       RMSE    &       0.188   &       0.190   &       0.135   &       0.139   &       0.108   &       0.089   \\  \hline
200     &       5       &       0.90    &       Bias    &       0.194   &       0.192   &       -1.196  &       -1.219  &       0.075   &       -0.051  \\
        &               &               &       RMSE    &       0.394   &       0.394   &       1.230   &       1.252   &       0.343   &       0.342   \\
200     &       25      &       0.90    &       Bias    &       0.168   &       0.170   &       -0.296  &       -0.308  &       0.069   &       -0.003  \\
        &               &               &       RMSE    &       0.264   &       0.264   &       0.370   &       0.379   &       0.216   &       0.208   \\
500     &       5       &       0.90    &       Bias    &       0.204   &       0.205   &       -1.183  &       -1.207  &       0.083   &       -0.035  \\
        &               &               &       RMSE    &       0.282   &       0.282   &       1.196   &       1.220   &       0.212   &       0.198   \\
500     &       25      &       0.90    &       Bias    &       0.161   &       0.164   &       -0.303  &       -0.314  &       0.064   &       -0.013  \\
        &               &               &       RMSE    &       0.203   &       0.205   &       0.331   &       0.340   &       0.133   &       0.130   \\  \hline
\end{tabular}
\vspace{3mm}
\caption{\emph{Small sample performance of panel quantile methods in Design 4.} \label{mcII}} 
\end{center}
\end{table}

\begin{table}
\begin{center}
\begin{tabular}{ cccccccccc } \hline
N & T & $\tau$ &  & \multicolumn{6}{c}{Quantile Regression Panel Data Methods}  \\
        &               &               &  &    QR      &       WQR     &       FE      &       WFE     &       PQR     &       WPQR        \\ \hline
200     &       5       &       0.50    &       Bias    &       0.232   &       0.258   &       0.001   &       -0.004  &       0.124   &       0.046   \\
        &               &               &       RMSE    &       0.239   &       0.269   &       0.061   &       0.071   &       0.135   &       0.082   \\
200     &       25      &       0.50    &       Bias    &       0.224   &       0.256   &       0.000   &       0.001   &       0.037   &       0.014   \\
        &               &               &       RMSE    &       0.226   &       0.259   &       0.028   &       0.035   &       0.045   &       0.037   \\
500     &       5       &       0.50    &       Bias    &       0.229   &       0.259   &       -0.001  &       -0.004  &       0.124   &       0.047   \\
        &               &               &       RMSE    &       0.232   &       0.263   &       0.038   &       0.047   &       0.128   &       0.064   \\
500     &       25      &       0.50    &       Bias    &       0.223   &       0.254   &       0.000   &       0.001   &       0.037   &       0.013   \\
        &               &               &       RMSE    &       0.224   &       0.255   &       0.017   &       0.021   &       0.040   &       0.025   \\  \hline
200     &       5       &       0.75    &       Bias    &       0.216   &       0.242   &       -0.063  &       -0.117  &       0.102   &       -0.024  \\
        &               &               &       RMSE    &       0.225   &       0.254   &       0.094   &       0.140   &       0.119   &       0.077   \\
200     &       25      &       0.75    &       Bias    &       0.210   &       0.241   &       -0.011  &       -0.021  &       0.029   &       -0.004  \\
        &               &               &       RMSE    &       0.213   &       0.245   &       0.033   &       0.042   &       0.040   &       0.036   \\
500     &       5       &       0.75    &       Bias    &       0.217   &       0.243   &       -0.064  &       -0.114  &       0.103   &       -0.023  \\
        &               &               &       RMSE    &       0.221   &       0.248   &       0.078   &       0.125   &       0.110   &       0.052   \\
500     &       25      &       0.75    &       Bias    &       0.211   &       0.238   &       -0.011  &       -0.021  &       0.029   &       -0.005  \\
        &               &               &       RMSE    &       0.212   &       0.240   &       0.022   &       0.031   &       0.033   &       0.023   \\  \hline
200     &       5       &       0.90    &       Bias    &       0.202   &       0.229   &       -0.292  &       -0.349  &       0.070   &       -0.070  \\
        &               &               &       RMSE    &       0.220   &       0.249   &       0.303   &       0.359   &       0.101   &       0.111   \\
200     &       25      &       0.90    &       Bias    &       0.203   &       0.230   &       -0.038  &       -0.062  &       0.019   &       -0.027  \\
        &               &               &       RMSE    &       0.207   &       0.234   &       0.054   &       0.077   &       0.037   &       0.053   \\
500     &       5       &       0.90    &       Bias    &       0.207   &       0.231   &       -0.292  &       -0.344  &       0.075   &       -0.066  \\
        &               &               &       RMSE    &       0.213   &       0.237   &       0.296   &       0.348   &       0.087   &       0.083   \\
500     &       25      &       0.90    &       Bias    &       0.201   &       0.228   &       -0.039  &       -0.062  &       0.018   &       -0.028  \\
        &               &               &       RMSE    &       0.203   &       0.230   &       0.045   &       0.068   &       0.027   &       0.038   \\ \hline
\end{tabular}
\vspace{3mm}
\caption{\emph{Small sample performance of panel quantile methods in Design 5.} \label{mcIV}} 
\end{center}
\end{table}

\subsection{Case 2: Conditional Missing at Random Models}
We now compare the performance of the proposed estimator WPQR in models with conditional missing data at random. We continue to employ model \eqref{kgm1}-\eqref{kgm3} to generate simulation data and expand the variants to the model to the following cases:
\begin{description}
\item[Design 3] We assume that $\rho_0 = \rho_1 = 0$ and $\theta_1 = \theta_2 = 1$. The error term of the main equation is assumed to be $\chi_3^2$ and $\alpha_i \sim \mathcal{N}(0,1)$. 
\item[Design 4] The design is similar to Design 3 but now we consider a missing data process which depends on the observed lagged value of the response variable. Then $\rho_0 = 0$, $\rho_1 = 0.5$ and $\theta_1 = \theta_2 = 0$. The error term $u_{it} \sim \chi_2^3$ and $\alpha_i \sim \mathcal{N}(0,1)$. 
\item[Design 5] We consider a simple variation of Design 4 by assuming $(u_{it},\alpha_i) \sim \mathcal{N}(\bm{0},\bm{I})$ and setting the intercept of \eqref{kgm2} equal to 5 in order to maintain the proportion of missing data.
\end{description}

We employ several sample sizes $ N =\{ 200, 500 \}$ and $T = \{ 5, 25 \}$ and compare the performance of the following estimators: (1) the pooled quantile regression estimator (QR); (2) a weighted version of the quantile regression estimator (WQR) as in Lipsitz et al. (1997) and Maitra and Vahid (2006); (3) Koenker's (2004) quantile regression estimator for a model with fixed effects (FE); (4) a weighted version of the quantile regression estimator for a model with fixed effects (WFE); (5) A penalized quantile regression estimator penalized estimator with $\lambda=1$ (PQR); (6) the penalized quantile regression estimator proposed in this study which uses propensity score weighting and $\lambda = 1$ (WPQR). We use the same weights for all the estimators. The weights $\hat{\pi}_{it}^{-1}$ are obtained after we estimate a logit model using the observed covariates and/or lagged independent variables as regressors.  

The results are presented in Tables \ref{mcI}, \ref{mcII}, and \ref{mcIV}. While Tables \ref{mcI} and \ref{mcII} show the bias and RMSE of the estimators in Designs 3 and 4 (when the conditional quantile function is non-linear), Table \ref{mcIV} shows the small sample performance of the estimators in Design 5 (when the conditional quantile function is linear under a missing data process that it is ignorable (MAR)). 

The tables consistently show that QR and WQR are biased in the presence of non-random missing data. The performance of the fixed effects estimators (FE and WFE) is satisfactory at the center of the conditional distribution of the response variable, but it deteriorates quickly as we move to the tails. In Table \ref{mcI} for instance, in a model with $N=200$ and $T=5$, the bias of the WFE estimator is -0.072 at the 0.5 quantile and -0.729 at the 0.9 quantile. In contrast, the bias of the WPQR estimator is relatively small and varies between 2\% and 4.5\%. We note that the WPQR estimator is not unbiased in the presence of individual effects that are correlated with the independent variables. The parameter $\pi=0.3$ in equation \eqref{kgm3}, and therefore, small biases are expected. The advantage of the simulation designs is that they allow us to see directly the advantages of shrinkage and the importance of $\lambda$ selection in models with endogenous covariates.

The relatively poor performance of the estimator proposed by Lipsitz et al. (1997) and Maitra and Vahid (2006) is not surprising since we consider the case of endogenous independent variables (i.e., $x_{it}$ and $\alpha_i$ are not independent). Moreover, the relatively poor performance of the fixed effects estimator is due to incidental parameters. Note that the bias decreases when $T$ increases, but it remains, in some cases, larger than 20\% when $T=25$ (i.e., Table \ref{mcII}). When the propensity score depends on the observed response variables, we continue to see that WPQR offers the best small sample performance in the class of panel quantile estimators. In the next section, we investigate the case of selection on unobservables.

\subsection{Case 3: Selection on unobservables}
This section expands the variants of the model by (i) considering a model where attrition depends on a variable that is not observed when the unit drops out of the sample, and (ii) considering a similar model to the model used in Kyriazidou (1997). The response variable is generated as in equations \eqref{kgm1} and \eqref{kgm2} but individual effects are generated by $\alpha_i = T_i^{-1} \sum_{t=1}^{T_i} x_{it} \pi_\alpha + \xi_{1,i}$ and $\eta_i = T_i^{-1} \sum_{t=1}^{T_i} w_{it} \pi_\eta + 2 \xi_{2,i}$, where $\eta_i = \alpha_i$ in equation \eqref{kgm2}, $\xi_{2,i}$ is distributed as uniform, and $(x_{it}, w_{it}) \sim \mathcal{N}(\bm{1},\bm{I})$. The distribution of the error term in equation \eqref{kgm1} is assumed to be Gaussian. The parameters $\beta_0 = 0$, $\beta_1=1$, $\pi_\alpha = \pi_\eta = 1$. The parameter of interest is $\beta_1$. We consider the following design: 
\begin{description}
\item[Design 6] We focus on a case where we have one possible ``refreshment" sample and $T=2$. We simulate mixed-continuous data between $T=1$ and $T=2$ and consider the sample closest to $T=2$ to satisfy Assumption \ref{A0}. We concentrate in the case of selection on unobservables by setting $\rho_0=0.5$ and $\rho_1=0$. 
\end{description}

Given the good performance of the proposed estimator in Tables \ref{mcI}-\ref{mcIV}, we concentrate our attention on the performance of the proposed estimator WPQR in the case of selection on unobservables. We expand the designs by considering several sample sizes $N = \{ 500, 2000 \}$ and $T = 2$ and compare the performance of the following estimators for the first stage: (1) the unfeasible estimator (UNF); (2) an estimator for the missing completely at random assumption (MCAR); (3) an estimator for the missing at random assumption (MAR); (4) an estimator for the assumption on selection on unobservables with a model estimated using refreshment samples (REF). 

Table \ref{mcV} shows two interesting findings. First, it is possible to improve the performance of the WPQR estimator in models with non-ignorable attrition by using a ``refreshment'' sample. Second, even for small $T$, the bias of the estimator is less than 10\% and the parameter is precisely estimated in comparison to the unfeasible estimator.  

\begin{table}
\begin{center}
%\fontsize{9}{10.4}\selectfont
\begin{tabular}{ c c c c c c c c c c c } \hline
\multicolumn{1}{l}{}&\multicolumn{2}{c}{}& \multicolumn{8}{c}{WPQR: First Stage Methods}  \\
\multicolumn{1}{l}{N}&\multicolumn{1}{c}{T}& \multicolumn{1}{c}{$\tau$}& \multicolumn{4}{c}{Bias}& \multicolumn{4}{c}{RMSE}  \\
\multicolumn{3}{l}{}&\multicolumn{1}{c}{UNF}&\multicolumn{1}{c}{MAR}&\multicolumn{1}{c}{MCAR}&\multicolumn{1}{c}{REF}&\multicolumn{1}{c}{UNF}&\multicolumn{1}{c}{MAR}&\multicolumn{1}{c}{MCAR}&\multicolumn{1}{c}{REF}\\  \hline
500     &       2               &       0.10    &       -0.015  &       -0.164  &       -0.107  &       -0.086  &       0.115   &       0.175   &       0.136   &       0.127   \\
500     &       2               &       0.25    &       -0.010  &       -0.186  &       -0.102  &       -0.081  &       0.088   &       0.194   &       0.123   &       0.108   \\
500     &       2               &       0.50    &       -0.003  &       -0.211  &       -0.089  &       -0.068  &       0.069   &       0.222   &       0.112   &       0.099   \\
500     &       2               &       0.75    &       0.005   &       -0.240  &       -0.079  &       -0.056  &       0.064   &       0.259   &       0.109   &       0.091   \\
500     &       2               &       0.90    &       0.011   &       -0.245  &       -0.073  &       -0.049  &       0.072   &       0.273   &       0.115   &       0.096   \\ \hline
2000    &       2       &       0.10    &       -0.005  &       -0.162  &       -0.094  &       -0.074  &       0.059   &       0.165   &       0.104   &       0.090   \\
2000    &       2       & 0.25  &       -0.001  &       -0.184  &       -0.091  &       -0.070  &       0.046   &       0.187   &       0.099   &       0.082   \\
2000    &       2       &       0.50    &       0.002   &       -0.214  &       -0.089  &       -0.068  &       0.040   &       0.217   &       0.097   &       0.078   \\
2000    &       2       &       0.75    &       0.001   &       -0.246  &       -0.084  &       -0.062  &       0.037   &       0.252   &       0.092   &       0.073   \\
2000    &       2       &       0.90    &       0.002   &       -0.271  &       -0.082  &       -0.060  &       0.037   &       0.283   &       0.094   &       0.074   \\ \hline
\end{tabular}
\vspace{3mm}
\caption{\emph{Small sample performance of the WPQR estimator. The columns describe the performance of the different first stage estimators} \label{mcV}} 
\end{center}
%\fontsize{11}{13.2}\selectfont
\end{table}

\section{An Empirical Application}\label{EmpApp}

A number of recent papers investigate the extent to which technology that enables communication between utility companies and consumers leads to higher electricity savings (Joskow 2012\nocite{joskow:12}, Harding and Lamarche 2016\nocite{hl16}, Harding and Sexton 2017\nocite{hs17}). These key developments originated by the development of a wide-spread introduction of time-of-use (TOU) pricing in the electricity sector and the increased availability of Big Data which enable consumers to use increasingly sophisticated devices to monitor and optimize their electricity usage.

The introduction of ``smart'' technologies and time-of-use pricing leads to new findings regarding the use of technologies that maximize consumers’ ability to respond to information on prices and quantity. The studies are typically based on a small number of households observed at high frequency over time. For instance, Harding and Lamarche (2016)\nocite{hl16} draw conclusions from a large scale randomized controlled trial of TOU pricing for residential consumption in a South Central US State. In their study, the electricity usage of 1011 households were recorded over 15-minute intervals, leading to a panel data set of more than 11 million observations. 

Despite the increasing popularity of empirical studies in this area (Jessoe and Rapson (2014)\nocite{katrina2014}, Ito (2014)\nocite{Ito2014}, Wolak (2011)\nocite{wolak:06}, among others), attrition has been ignored in the empirical literature. Households move, drop out of the sample for unknown reasons and/or can request changes in the technology randomly assigned to them due to incompatibility to different settings. Although some experiments have a high degree of compliance among treated participants, it is possible that a number of participants are switched to alternative treatments because issues with the installed technology. Naturally, missing data can create estimation issues associated with the use of non-random samples over time, even though the data is likely to be obtained from a reliable allocation of households to treatment groups and control groups at the beginning of the randomized trial period. To investigate attrition due to latent variables in this setting, we generate a simulation experiment using data from electricity consumption in Ireland. We generate different levels of attrition and investigate the performance of existing methods and the proposed approach. 

\begin{table}
\begin{center}
\begin{tabular}{ l c c c c  } \hline
\multicolumn{1}{c}{Variables}&\multicolumn{2}{c}{Control}& \multicolumn{2}{c}{Treatment} \\ 
\multicolumn{1}{c}{}&\multicolumn{1}{c}{Mean}&\multicolumn{1}{c}{Std Dev}&\multicolumn{1}{c}{Mean}&\multicolumn{1}{c}{Std Dev} \\ \hline
Electricity usage at 6 AM       (Night)&        0.101   &       0.139   &       0.109   &       0.164   \\
Electricity usage at 6 PM       (Peak) &        0.415   &       0.463   &       0.409   &       0.470   \\
Electricity usage at 8 PM       (Day)  &        0.387   &       0.408   &       0.406   &       0.433   \\
Household size  &       0.251   &       0.434   &       0.340   &       0.474   \\
One or more adults at home      &       0.647   &       0.478   &       0.690   &       0.462   \\
One or more kids at home        &       0.147   &       0.354   &       0.165   &       0.371   \\
Electric heater &       0.060   &       0.237   &       0.070   &       0.255   \\
Electric cook   &       0.721   &       0.448   &       0.670   &       0.470   \\
Head of household employed      &       0.538   &       0.499   &       0.620   &       0.485   \\
House size      &       0.426   &       0.494   &       0.460   &       0.498   \\
Insulated attic &       0.885   &       0.319   &       0.895   &       0.307   \\
Insulated walls &       0.555   &       0.497   &       0.600   &       0.490   \\
Age of the house $<$ 10 years   &       0.181   &       0.385   &       0.170   &       0.376   \\
Age of the house 10 to 30 years &       0.221   &       0.415   &       0.310   &       0.462   \\
Temperature     &       6.537   &       5.386   &       6.537   &       5.386   \\
Relative humidity       &       93.317  &       4.349   &       93.317  &       4.349   \\ \hline
Number of Households & 470 & & 200 & \\
``Population" at hour $h$       &       170,610 &               &       72,600  &               \\
``Population"   &       8,189,280       &               &       3,484,800       &               \\ \hline
\end{tabular}
\vspace{3mm}
\caption{\emph{Descriptive statistics.} \label{table41}} 
\end{center}
\end{table}

\subsection{Data}

The data employed in this paper is obtained from a large scale randomized control trial as part of Irelands' smart metering plan for residential electricity consumption. The smart meter data is obtained from the Irish Social Science Data Archive (ISSDA) and we use the CER Smart Metering Project. The data set used in this paper consists of household smart meter readings measured over 30 minute intervals for $N=670$ households. The large scale experiment was conducted from 2008 to 2011 and we employ data from the period January 2010 to December 2010. The period June 2009 to December 2009 is the period before the implementation of the policy. In this period, baseline data was collected and the participants were assigned into treatment and control groups. In the second period, from January 2nd, 2010 to December 31st, 2010, the electricity usage of households in the control group was recorded as well as the electricity consumption of the households in the treatment groups. 

The participants of the program were selected to ensure an adequate representation of the national population. They were assigned to two treatment types. First, treated customers were charged at different rates during weekdays: Tariff A is 12 cents per kilowatt hours (kwh) from 23:00 to 8:00 (Night), 14 cents per kwh from 8:00 to 23:00 (Day) with the exception of 17:00 to 19:00, and 20 cents per kwh from 17:00 to 19:00 (Peak); Tariff B is 11 cents per kwh, 13.5 cents per kwh and 26 cents per kwh; Tariff C is 10 cents per kwh, 13 cents per kwh and 32 cents per kwh; and Tariff D is 9 cents per kwh, 12.5 cents per kwh and 38 cents per kwh, respectively. The rates are in Euro cents and they exclude a consumption tax (value added tax). In this study, we concentrate our attention on Tariff B. The control group has a time invariant rate of 14.1 cents per kwh.

The second treatment relates to the enabling technology. There are three treatment groups: Monthly bill combined with an energy usage statement (T1); Bimonthly billing combined with an energy usage statement plus overall load reduction (T2); in-home display (IHD) device as well as a Bimonthly billing combined with an energy usage statement (T3). An IHD is a small wireless device which displays information on electricity usage and costs in real time. We proceed in this study creating a treatment group for households in these groups. The control group receives bimonthly electricity bills. 

The dependent variable is electricity consumption, measured in kilowatt hours, at the residential level (Table \ref{table41}). The control group includes 470 households and the other 200 households were assigned to the different treatments. The data includes information on the average temperature in Ireland, average relative humidity, an indicator for household size (4 or more people in the home), an indicator for electricity used to heat home (either central or plug in), an indicator for electric stove for cooking, an indicator variable for whether the head of the household is employed, an indicator for the size of the house (e.g., 3 rooms or more rooms), and indicators for the characteristics of the house. 

\subsection{Model}
Because TOU pricing vary by hour, we estimate the treatment effect at hour $h$ of the day corresponding to the three different tariffs: Night, Peak and Day. This approach is consistent with existing models of electricity consumption, most notably Ramanathan, Engle, Granger, Vahid-Araghi, and Brace (1997). To model electricity consumption, we follow the model first suggested by Ramanathan et al. (1997)\nocite{ramanathan:97} and model the load function as piecewise constant over the interval of time, $h$, for which electricity consumption is measured. This gives rise to the following equation:
\begin{equation}
\log (Y_{i,t,h})=\beta_{0,h} + \delta_h d_{i,h} + \bm{x}_{i}' \bm{\beta}_{1,h} + f(W_{t,h}) + \alpha_{i,h} + \epsilon_{i,t,h}, \label{mainE}
\end{equation} 
where $i= 1, \hdots, N$ denote households, $t = 1, \hdots, T$ denote days, and $f(W_{t,h})$ corresponds to a smooth function of weather measurement that can be generated as univariate splines of temperature and relative humidity. Our quantile treatment coefficients are identified by comparing electricity usage in the control group to that in the treated group:
\begin{equation}
Q_{\log (Y_{i,t,h})}(\tau | d_{i,h}, \bm{x}_{i}, W_{t,h}) = \beta_{0,h}(\tau) + \delta_{h}(\tau) d_{i,h} + \bm{x}_{i}' \bm{\beta}_{1,h}(\tau) + f(W_{t,h};\tau) + \alpha_{i,h}(\tau), \label{mainQ}
\end{equation}
where $\delta_{h}(\tau)$ is the quantile treatment effect (QTE) of interest. For the purpose of the simulation experiment performed in this section, we are interested in investigating how the households that were using an in-home display (IHD) device plus a Bimonthly billing combined with an energy usage statement compare with the control group when pricing have consirable changes over time. 

\begin{table}
\begin{center}
\begin{tabular}{ l c c c c c c } \hline
Variable & \multicolumn{5}{c}{Quantile Regression} & Mean \\
                     &  0.10 & 0.25 & 0.50 & 0.75 & 0.90        \\ \hline
& \multicolumn{6}{c}{12 months period} \\                               
Treatment at 6 AM & -0.021 & -0.028 & -0.010 & 0.001 & 0.008 & -0.001 \\
\hspace{1mm}    (11 cents per kwh)  &(0.007)&(0.005)&(0.004)&(0.004)&(0.007)&(0.004) \\
Treatment at 6 PM & -0.071 & -0.075 & -0.107 & -0.100 & -0.043 & -0.067 \\
\hspace{1mm}    (26 cents per kwh)  &(0.009)&(0.007)&(0.006)&(0.006)&(0.007)&(0.005) \\
Treatment at 8 PM & -0.080 & -0.056 & -0.064 & -0.025 & 0.021 & -0.031 \\
\hspace{1mm}    (13.5 cents per kwh)  &(0.009)&(0.005)&(0.005)&(0.006)&(0.007)&(0.005) \\ \hline
Controls & Yes & Yes & Yes & Yes & Yes & Yes \\
Observations per household & 363 & 363 & 363 & 363 & 363 & 363 \\
Size of the ``population" & 243,210 & 243,210 & 243,210 & 243,210 & 243,210 & 243,210 \\ \hline
& \multicolumn{6}{c}{First two months} \\                                                                                                               
Treatment at 6 AM&0.021&-0.044&-0.044&-0.055&-0.041&-0.032\\
\hspace{1mm}    (11 cents per kwh) &(0.015)&(0.012)&(0.009)&(0.012)&(0.015)&(0.010) \\
Treatment at 6 PM&-0.078&-0.130&-0.143&-0.101&-0.008&-0.085\\
\hspace{1mm}    (26 cents per kwh) &(0.025)&(0.016)&(0.013)&(0.015)&(0.014)&(0.012) \\
Treatment at 8 PM&-0.090&-0.065&-0.085&-0.029&0.017&-0.052\\
\hspace{1mm}    (13.5 cents per kwh) &(0.019)&(0.012)&(0.010)&(0.012)&(0.016)&(0.010) \\ \hline
Controls&Yes&Yes&Yes&Yes&Yes&Yes\\
Observations per household&59&59&59&59&59&59\\
Size of the ``population" &39,530&39,530&39,530&39,530&39,530&39,530\\  \hline
\end{tabular}
\vspace{3mm}
\caption{\emph{Population regressions by OLS and Quantile Regression. Standard errors are in parentheses.} \label{table42}} 
\end{center}
\end{table}

Table \ref{table42} presents results obtained from estimating equation \eqref{mainQ} separately for each hour $h$ of the day. The upper block of the table shows the quantile treatment effect at different TOU pricing for the period January to December 2010. The lower block of the table shows evidence on the short term effects of the pricing policy, as we restrict the sample to include observations over the first two months after the implementation of the policy. The results show that, as expected, both the introduction of ``smart'' technology and time-of-use pricing lead to a reduction of electricity consumption with the largest gain in the first months after the implementation of the program. At the mean level, we find a reduction of 8.1\% at the peak hour in the first two months, 5.1\% at 8 PM and 3.1\% at 6 AM. At peak hours, there appear to be considerably heterogeneity across  quantiles. While the treatment effect is -12.2\% at the 0.25 quantile of the conditional distribution of electricity consumption, it is not statistically significantly different than zero at the 0.9 quantile of the conditional distribution. Another interesting, yet expected finding is that, in general, the QTE estimates at the peak hour are smaller than the QTE estimates during day hours, when the price is reduced by 50\%. (The sole exception is $\tau=0.10$).

\subsection{Simulation experiment}
The simulation experiment is based on the sample of households described in Table \ref{table41}. These households are considered to be the ``population". Following closely Bhattacharya (2008), the simulation exercise is performed as follows. First, we estimate equation \eqref{mainQ} and we treat these estimates as population parameters. The regression results are presented in the lower block of Table \ref{table42}. We did not include weather variables for simplicity but we include the 11 covariates shown in Table \ref{table41}. From this population, we generate an artificial electricity usage variable, $\log (Y_{i,t,h}) = \hat{\beta}_{0,h}(\tau) + \hat{\delta}_{h}(\tau) d_{i,h} + \bm{x}_{i,h}' \hat{\bm{\beta}}_{1,h}(\tau) + \alpha_{i,h} + u_{i,t,h}$, where $\alpha_{i,h} = d_{i,h} \xi_{1,i,h} + \sqrt{0.5} \xi_{2,i,h}$. The variables $(\xi_{1,i,h},\xi_{2,i,h},u_{i,t,h})$ are distributed as independent standard normal. 

Second, we generate attrition for this population considering the following equation for the missing data process:
\begin{equation}
s_{i,t,h} = 1 \{ \rho_0 \log (Y_{i,t,h}) + \rho_1 \log (Y_{i,t-1,h}) - v_{i,t,h} > 0 \}, \label{ee:simatt}
\end{equation}
where $v_{i,t,h}$ is a normally distributed random variable with mean 5 and variance 1, which gives the case of no attrition when $\rho_0=\rho_1=0$. To consider a model with selection on unobservables, we assume $\rho_0 \in \{0.5,1\}$. We also consider the case of no attrition, say $\rho_0=0$, and compare the performance of the methods. The parameter $\rho_1$ controls the degree of ``ignorable" selection and it is set to zero to concentrate on selection on unobservables. Notice that the $\alpha_{i,h}$'s can be a source of attrition, although they are centered at zero and are correlated with the treatment variable. 

Third, we estimate equation \eqref{mainQ} for the artificial electricity usage variable considering estimators for models with attrition: WQR as in Lipsitz et al. (1997) and Maitra and Vahid (2006) and our proposed estimator WPQR. For comparison on the performance of the estimator of the first stage, we include the unfeasible version which uses $\pi_{0,i,t}(h) = P(s_{i,t,h} = 1 | \log(y_{i,t,h}))$ and the feasible version which estimates the propensity score using an i.i.d. sample for the missing data and employs parametric methods. The ``refreshment" sample is obtained from the ``population" to satisfy the condition $P(s_{i,t,h} = 1 | \log(y_{i,t,h})) = P(s_{i,t,h} = 1 | \log(y_{i,t,h'}))$ a.s., where $h' \in [h - \epsilon,h+\epsilon]$ is defined as the closest 30-minute interval to $h$ within the same TOU tariff bracket. We repeat the exercise 400 times and we obtain the average estimate of the QTE and root mean square error (RMSE) of existing approaches and the proposed methods. 

\subsection{Streaming sample}

We used ``refreshment" sample in quotations to indicate that while the idea is to use additional data as in additive-non-ignorable models (i.e., Hirano et al. 2001), we do not have a fresh new sample of subjects in the second period. In the current setup it is typical for the sensors to receive continuous recordings of electricity consumption. While the devices record this data (and it is often stored, though sometimes also discarded) not all the data is used for analytics. In our case  these additional data are not used for identification and estimation of the quantile treatment effects. To avoid confusion with ``refreshment" samples, we refer to ``streaming" measures or ``streaming" sample. This is a term borrowed from the computer science and engineering literature on the use of Big Data methods and denotes the fact that it is common for the sensors recording the data to receive ``streaming data" (potentially at a very fast rate) but that in practice only a small subsample of all the available data stream is used to conduct the analysis. If needed and given that it was previously stored, systems are in place to retried the additional data. While this type of data may not be that familiar to economists, it is quite common in practice (see, e.g., Babcok et al. 2002\nocite{babcock02}, Hofleitner et al. 2012\nocite{Hofleitner12}, and Moreira-Matias et al. 2013\nocite{moreira2013}).

\begin{figure}
\begin{center}
\centerline{\includegraphics[width=.8\textwidth]{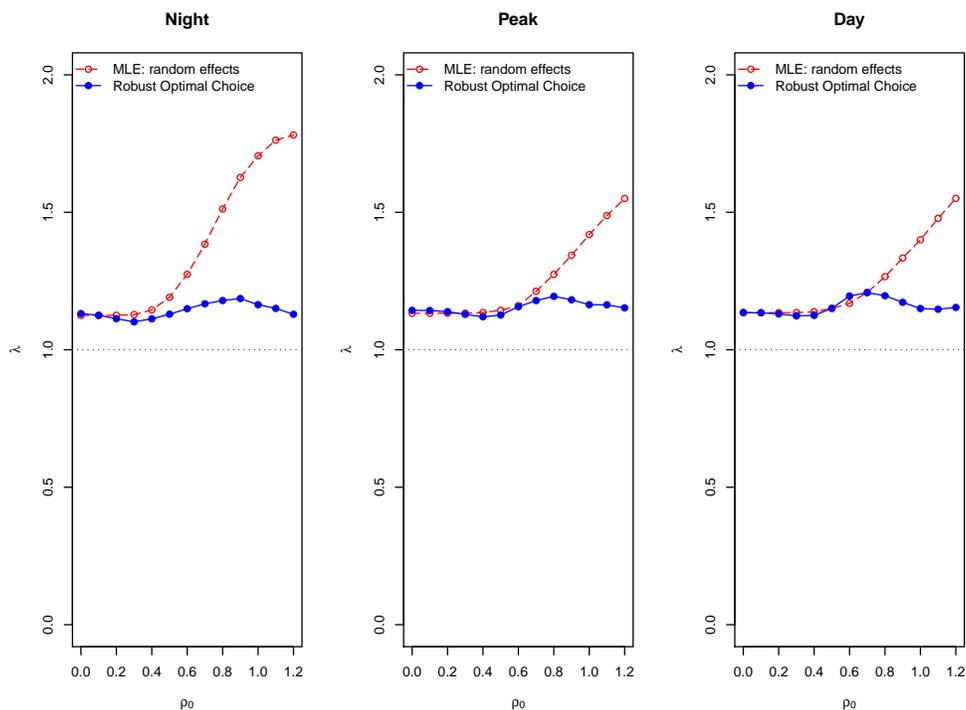}}
\caption{\emph{Tuning parameter selection in a model with attrition.}  \label{ee-fig2}}
\end{center}
\end{figure}

\subsection{The choice of $\lambda$}
We propose to select $\lambda$ in this empirical application by adapting a procedure first proposed by Belloni and Chernozhukov (2011)\nocite{belloni2011} to the case of missing observations. We introduce the following random variable:
\begin{equation*}
\Lambda_h = T \max_{1 \leq j \leq p+N} \left| \frac{1}{NT} \sum_{t=1}^T \sum_{i=1}^N \frac{s_{i,t}}{ \hat{\pi}_{i,t} } \frac{X_{j,i,t,h}}{\hat{\sigma}_{j,h}} \frac{ [\tau - I( \utilde{u}_{i,t,h} \leq \tau)] }{ \sqrt{ \tau (1 - \tau)}} \right| 
\end{equation*}
where the random variables $\utilde{u}_{1,1,h},\utilde{u}_{1,2,h},\hdots,\utilde{u}_{N,T,h}$ are i.i.d. $\mathcal{U}(0,1)$ independent of $X_{j,i,t,h}$ and $\sigma_{j,h}$ is the standard deviation of the variable $X_{j,i,t,h}$, which is the $j$-th covariate of the vector $\bm{X}_{i,t,h} = (d_{i,h},\bm{x}_i',\hat{f}(W_{t,h};\tau),\bm{z}_i')'$. Then, we set $\lambda_h = \varkappa \cdot \Lambda_h(1 - c)$, where $\varkappa = \kappa \sqrt{ \tau (1 - \tau)}$, the constant $\kappa > 0$, and $\Lambda_h(1 - c)$ is the $(1-c)$-quantile of the random variable $\Lambda_h$ conditional on $\bm{X}_{i,t,h}$. Following practical recommendations, we set $\kappa=2$ and $c=0.1$. The procedure is robust to different levels of attrition since the correction for missing data can be easily accommodated  by considering $\hat{\pi}_{i,t}$. Therefore, we label the procedure robust optimal choice of $\lambda$.

The robust approach was contrasted with two existing $\lambda$ selection methods. While Chen, Wan and Zhou (2015)\nocite{xChen2015} proposed a $K$-fold cross-validation procedure for cross-sectional regression with missing observations, Koenker (2005) proposes to estimate $\lambda$ by $\hat{\lambda} = \hat{\sigma}_u / \hat{\sigma}_\alpha$, where $\sigma_u^2$ is the variance of the error term and $\sigma_\alpha^2$ is the variance of the individual effect. The estimation of $\lambda$ works well in Gaussian models under non-missing data and it can be accomplished by employing standard maximum likelihood methods for random effects models. In a problem with missing data, the estimator $\hat{\lambda}$ is similarly defined but it is obtained by estimating the variance of $u$ and $\alpha$ using observed data. 

Figure \ref{ee-fig2} shows the average value of the selected $\lambda$ parameter over 200 random samples. The figure shows two estimated values for $\lambda$, the MLE estimator $\hat{\lambda}$ above and the selected $\lambda$ parameter using the robust procedure proposed in this paper. The value of $\rho_0$ generates different degrees of attrition in our application, ranging from 75 percent during night hours to over 20 percent during peak and day hours. We find that the performance of the estimator for $\lambda$ is quite satisfactory, in particular for models with a large degree of attrition. The value of $\lambda$ is does not vary with different levels of attrition, in contrast with the estimates obtained from MLE random effects estimator. The $K$-fold cross-validation approach did not outperform our preferred robust approach and is not illustrated here.     

\begin{table}
\begin{center}
\begin{tabular}{ c c c c c c c c c c c c } \hline
  &             &  & & &             & \multicolumn{3}{c}{Treatment effect} & \multicolumn{3}{c}{RMSE}\\      
                
$\tau$  & $\rho_0$ &    $NT$  &       Attrition  & $\delta_h$ & $\lambda$  & WQR   &       \multicolumn{2}{c}{WPQR} &       WQR     &       \multicolumn{2}{c}{WPQR}\\      
        &          &        &     &                  &        &       &                UNF    &       STR         & &              UNF    &       STR  \\ \hline
                     &   &       & \multicolumn{9}{c}{Night (11 cents per kwh)} \\           \hline
0.1     &       0.0     &       39530   &       0.000   &       0.021   &       0.581   &       -0.428  &       0.025   &       0.025   &       0.465   &       0.131   &       0.131   \\
0.1     &       0.5     &       34551   &       0.126   &       0.021   &       0.622   &       -0.181  &       0.027   &       0.029   &       0.226   &       0.126   &       0.125   \\
0.1     &       1.0     &       5888    &       0.851   &       0.021   &       0.776   &       0.530   &       0.053   &       0.050   &       0.524   &       0.116   &       0.116   \\
0.5     &       0.0     &       39530   &       0.000   &       -0.044  &       1.127   &       -0.044  &       -0.044  &       -0.044  &       0.102   &       0.130   &       0.130   \\
0.5     &       0.5     &       38202   &       0.034   &       -0.044  &       1.127   &       0.018   &       -0.043  &       -0.043  &       0.119   &       0.128   &       0.128   \\
0.5     &       1.0     &       15336   &       0.612   &       -0.044  &       1.161   &       0.453   &       -0.006  &       -0.008  &       0.506   &       0.120   &       0.117   \\
0.9     &       0.0     &       39530   &       0.000   &       -0.041  &       0.940   &       0.406   &       -0.039  &       -0.039  &       0.463   &       0.127   &       0.127   \\
0.9     &       0.5     &       39200   &       0.008   &       -0.041  &       0.817   &       0.416   &       -0.038  &       -0.038  &       0.473   &       0.127   &       0.127   \\
0.9     &       1.0     &       28081   &       0.290   &       -0.041  &       0.805   &       0.580   &       0.022   &       0.023   &       0.633   &       0.134   &       0.135   \\ \hline
                  &      &       & \multicolumn{9}{c}{Peak (26 cents per kwh)} \\         \hline
0.10    &       0.00    &       39530   &       0.000   &       -0.078  &       0.611   &       -0.526  &       -0.074  &       -0.074  &       0.465   &       0.129   &       0.129   \\
0.10    &       0.50    &       38144   &       0.035   &       -0.078  &       0.591   &       -0.413  &       -0.072  &       -0.071  &       0.355   &       0.130   &       0.130   \\
0.10    &       1.00    &       17569   &       0.556   &       -0.078  &       0.711   &       0.174   &       -0.023  &       -0.022  &       0.265   &       0.120   &       0.118   \\
0.50    &       0.00    &       39530   &       0.000   &       -0.143  &       1.144   &       -0.143  &       -0.141  &       -0.141  &       0.103   &       0.130   &       0.130   \\
0.50    &       0.50    &       39368   &       0.004   &       -0.143  &       1.126   &       -0.131  &       -0.141  &       -0.141  &       0.103   &       0.130   &       0.130   \\
0.50    &       1.00    &       33054   &       0.164   &       -0.143  &       1.125   &       0.114   &       -0.133  &       -0.132  &       0.272   &       0.124   &       0.123   \\
0.90    &       0.00    &       39530   &       0.000   &       -0.008  &       0.942   &       0.439   &       -0.008  &       -0.007  &       0.463   &       0.128   &       0.128   \\
0.90    &       0.50    &       39508   &       0.001   &       -0.008  &       0.888   &       0.440   &       -0.007  &       -0.007  &       0.464   &       0.127   &       0.127   \\
0.90    &       1.00    &       38416   &       0.028   &       -0.008  &       0.804   &       0.479   &       -0.002  &       -0.001  &       0.502   &       0.125   &       0.125   \\  \hline
                   &     &       & \multicolumn{9}{c}{Day (13.5 cents per kwh)} \\           \hline
0.10    &       0.00    &       39530   &       0.000   &       -0.090  &       0.645   &       -0.539  &       -0.088  &       -0.088  &       0.465   &       0.129   &       0.129   \\
0.10    &       0.50    &       38626   &       0.023   &       -0.090  &       0.590   &       -0.450  &       -0.086  &       -0.085  &       0.378   &       0.130   &       0.129   \\
0.10    &       1.00    &       20975   &       0.469   &       -0.090  &       0.641   &       0.105   &       -0.039  &       -0.039  &       0.211   &       0.121   &       0.117   \\
0.50    &       0.00    &       39530   &       0.000   &       -0.085  &       1.148   &       -0.084  &       -0.084  &       -0.084  &       0.103   &       0.130   &       0.130   \\
0.50    &       0.50    &       39374   &       0.004   &       -0.085  &       1.125   &       -0.073  &       -0.084  &       -0.084  &       0.103   &       0.130   &       0.130   \\
0.50    &       1.00    &       33017   &       0.165   &       -0.085  &       1.120   &       0.162   &       -0.079  &       -0.078  &       0.262   &       0.123   &       0.123   \\
0.90    &       0.00    &       39530   &       0.000   &       0.017   &       0.939   &       0.464   &       0.018   &       0.018   &       0.463   &       0.127   &       0.127   \\
0.90    &       0.50    &       39504   &       0.001   &       0.017   &       0.876   &       0.465   &       0.018   &       0.018   &       0.464   &       0.127   &       0.127   \\
0.90    &       1.00    &       38138   &       0.035   &       0.017   &       0.832   &       0.510   &       0.028   &       0.028   &       0.507   &       0.127   &       0.126   \\ \hline
\end{tabular}
\vspace{3mm}
\caption{\emph{Performance of Weighted Quantile Regression Estimators Under Selection on Unobservables. The unfeasible panel data estimator is denoted by UNF and the feasible version using a ``streaming" sample is denoted by STR.} \label{table43}} 
\end{center}
\end{table}

\begin{figure}
\begin{center}
\centerline{\includegraphics[width=.8\textwidth]{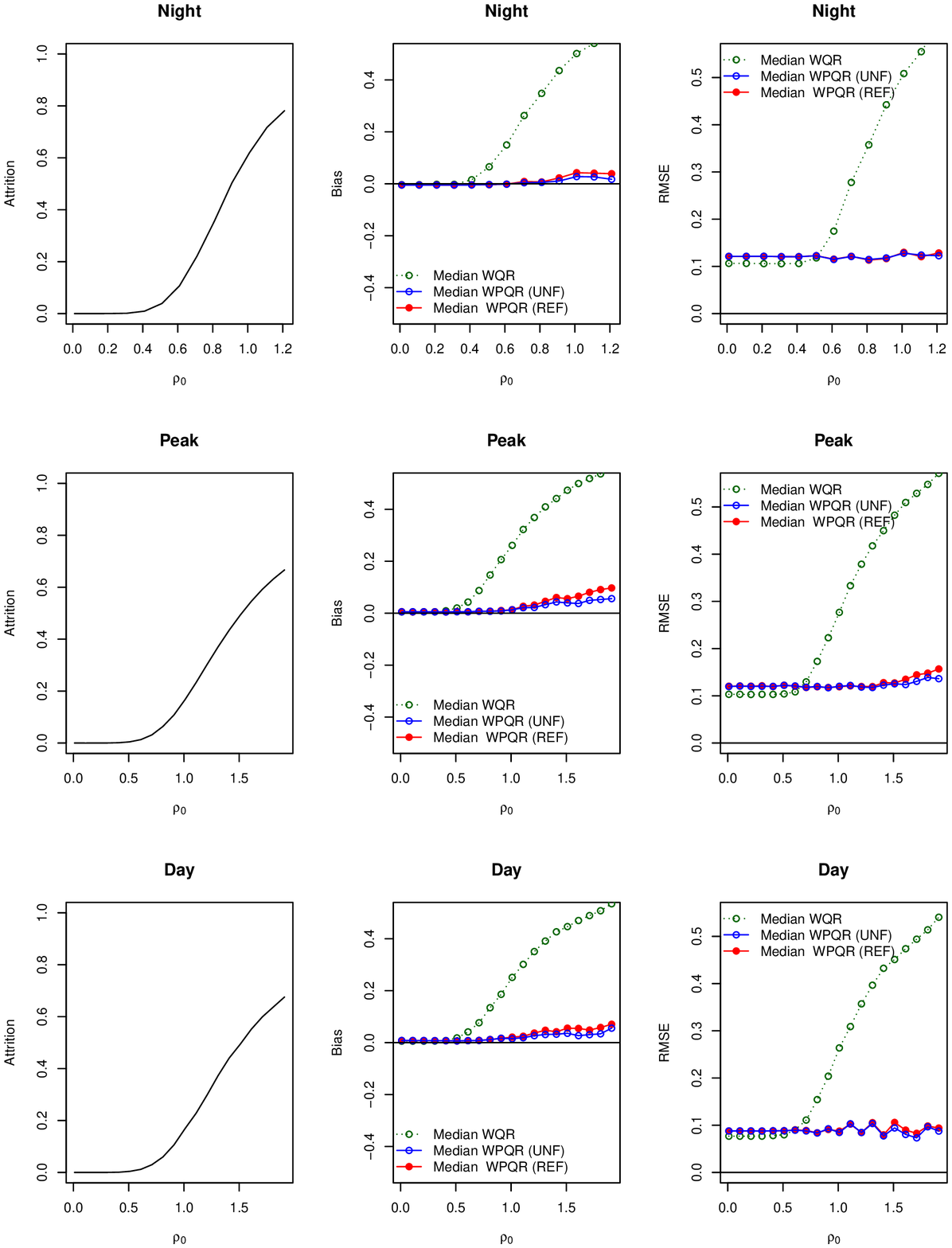}}
\caption{\emph{ Attrition and its impact on the estimation of TOU pricing.}  \label{ee-fig1}}
\end{center}
\end{figure}

\subsection{Empirical Results}

Table \ref{table43} reports results for the QTE, $\delta_h(\tau)$, at $\tau \in \{0.1,0.5,0.9\}$ and $\rho_0 \in \{0,0.5,1\}$. It also shows the proportion of missing observations, the value of the parameter of interest $\delta_h$, and the value of the selected tuning parameter, $\lambda$. We expect the results to deteriorate as $\rho_0$ increases, in particular at the tails of the conditional distribution. We interpret the difference between WPQR and WQR as estimates with and without correction for selection on unobservables. The comparison of the QTE estimates obtained previously in Table \ref{table42} and the results shown in the first columns of Table \ref{table43} illustrate the advantage of the proposed approach. The corrected estimates perform better than the estimates not corrected for selection on unobservables in a panel quantile model. It is also interesting to see that the correction seem to perform well under different values of the treatment effect coefficient which varies by quantile and TOU tariff. 

Figure \ref{ee-fig1} presents additional results. For simplicity in exposition, we provide evidence on the performance of two estimators (WQR, and WPQR) at the median quantile $\tau=0.5$. The estimator labeled WPQR (UNF) is the unfeasible version of the estimator defined in equation \eqref{wpqr} using the weights obtained from knowing the parameters of equation \eqref{ee:simatt}. On the other hand, WPQR (STR) is the feasible version of the estimator. The graph presents attrition, bias and RMSE. As expected, the number of observations, $NT$, decreases when $\rho_0$ increases. 

The results show that the bias changes as $\rho_0$ increases, although the bias of the WQR estimator monotonically increases at a faster rate. Naturally, the RMSE of the estimator seems to increase on $\rho_0$ too. At the 0.5 quantile, the performance of the panel quantile methods is similar, exhibiting small biases when the proportion of non-random missing observations is between 0\% and 60\%. We do see that the proposed approach offers the best performance and tends to provide smaller biases and RMSE for larger values of $\rho_0$. It is interesting to see that the performance of the feasible version of the estimator compares quite well to the performance of the unfeasible estimator. 

\subsection{Missing Covariates}

In the evaluation of electricity pricing experiments using ``streaming" data, the assumption that covariates are observed is easily verifiable in the data. The model estimated in equation (4.2) include (i) baseline characteristics and (ii) ``external" covariates that are not household-specific. The vector of hour-of-day invariant variables $\bm{x}_i$ includes indicators for household size and the size of the house, an indicator for electricity used to heat home, an indicator for electric stove for cooking, an indicator variable for whether the head of the household is employed, indicators for the characteristics of the house, and indicators for the age of the house (i.e., age $\leq 10$ years and age between 10 and 30 years old). The external covariates are temperature and relative humidity in Ireland, and we employ cubic B-spline basis functions to estimate the function $f(W_{t,h})$. Lastly, the treatment indicator $d_{i,h}$ is observed because it is a determinist function of the time of the day.

In other applications, however, the assumption that treatment variables and controls are observed for all time periods seems unlikely. There are several approaches discussed in the literature for the estimation of conditional mean models (see Robins and Wang (2000)\nocite{robins2000}, Roy and Lin (2002)\nocite{roylin2002}, D'Agostino and Rubin (2000)\nocite{RubinDonald2000}, among others), in contrast to the quantile regression literature that remains largely undeveloped. The sole exceptions are Wei, Ma and Carroll (2012)\nocite{WeiMaCarroll2012} and Wei and Yang (2014)\nocite{WeiYang2014} but their approaches are designed to address missing covariates in cross-sectional data. More importantly, they assume that the response variable is observed for all subjects, which is likely to be violated in applications in panel data. A general approach for the case of missing treatments and covariates is out of the scope of this paper and it requires further investigation. D\'iaz (2017)\nocite{diaz2017} seems a good starting point but the approach is not developed for panel data.

\section{Conclusions}

Non-random attrition in randomized field trials, as originally pointed out by Hausman and Wise (1979), raises several issues in panel data. Only a few papers investigate this issue in quantile regression, but they require that unobserved individual heterogeneity to be independent of the independent variables and the methods only addresses issues associated with selection on observables. These assumptions are typically considered to be strong for the analysis of large randomized field trials. These studies include recent Time-of-Day electricity pricing experiments inspired by the work of Aigner and Hausman (1980). 

This paper introduces a quantile regression estimator for panel data models with individual heterogeneity and attrition. The method is motivated by the fact that attrition bias is often encountered in Big Data problems. Our paper however makes two distinct contributions to the existing literature. First, we propose a method to estimate a model with individual unobserved heterogeneity that can be a source of attrition. Second, our method exploits additional data obtained by the increased availability of Big Data of households' panels. The estimator is computationally easy to implement in Big Data applications with a large number of subjects.  We investigate the conditions under which the parameter estimator is asymptotically Gaussian and we carry out a series of simulations to investigate the finite sample properties of the estimator.

\bibliographystyle{econometrica}
\bibliography{qatt}

\appendix
\section{Proofs} 
Throughout this appendix, we omit $\tau$ in $\bm{\theta}(\tau)$ for notational simplicity and the proofs refer to Knight's (1998)\nocite{kK98} identity. If we denote the quantile influence function by $\psi_{\tau}(u) = \tau - I(u \leq 0)$, for $u \neq 0$, $\rho_{\tau}(u-v) - \rho_{\tau}(u) = - v \psi_{\tau} + \int_{0}^{v} (I(v \leq s) - I(v \leq 0)) ds$. 

\begin{lemma}\label{L1}
Let $S(\bm{\theta},\pi(\bm{\gamma})) = S(\bm{\theta},\bm{\gamma})$ be a convex function in $\bm{\theta}$. Assume that $\sup_{\bm{\theta}} | S(\bm{\theta},\hat{\bm{\gamma}}) - S(\bm{\theta},\bm{\gamma}_0) | = o_p(1)$. For any $\epsilon > 0$, let $S(\hat{\bm{\theta}}_1,\hat{\bm{\gamma}}) < \inf_{\| \bm{\theta} - \hat{\bm{\theta}}_1 \|} S(\bm{\theta},\hat{\bm{\gamma}})$ and $S(\hat{\bm{\theta}}_2,\bm{\gamma}_0) < \inf_{\| \bm{\theta} - \hat{\bm{\theta}}_2 \|} S(\bm{\theta}, \bm{\gamma}_0)$. Then, $\|  \hat{\bm{\theta}}_1 - \hat{\bm{\theta}}_2 \| = o_p(1)$.
\end{lemma}
\begin{proof}
See Lemma 2 in Galvao, Lamarche and Lima (2013).
\end{proof}

\begin{proof}[Proof of Proposition \ref{P1}]
The result is shown along the lines of Wooldridge (2007). We write
\begin{eqnarray*}
E( M_{it}(\bm{\theta},\pi_{0})) & = & E \left\{ [s_{it}/\pi_{0,it}] \bm{X}_{it} \psi_{\tau} (Y_{it} - \bm{X}_{it}' \bm{\theta}) - \lambda \bm{Z}_{i} \psi_{\tau}(\bm{z}_i' \bm{\alpha})   \right\} \\
 & = & E \left\{ E \left( [s_{it}/\pi_{0,it}] \bm{X}_{it} \psi_{\tau} (Y_{it} - \bm{X}_{it}' \bm{\theta}) - \lambda \bm{Z}_{i} \psi_{\tau}(\bm{z}_i' \bm{\alpha})  \;\middle\vert\; \bm{X}_{it}  \right) \right\} \\
& = & E \left\{ E \left( [s_{it}/\pi_{0,it}] \bm{X}_{it} \psi_{\tau} (Y_{it} - \bm{X}_{it}' \bm{\theta})  \;\middle\vert\; \bm{X}_{it}  \right) - E \left( \lambda \bm{Z}_{i} \psi_{\tau}(\bm{z}_i' \bm{\alpha})  \;\middle\vert\; \bm{X}_{it}  \right) \right\} \\
& = & E \left\{ E \left( [s_{it}/\pi_{0,it}] \bm{X}_{it} \psi_{\tau} (Y_{it} - \bm{X}_{it}' \bm{\theta})  \;\middle\vert\; \bm{X}_{it}  \right) \right\}, 
\end{eqnarray*}
where the last equality was obtained by Assumption \ref{A3}. Moreover, Assumptions \ref{A0} and \ref{A3} imply that,
\begin{eqnarray*}
E( M_{it}(\bm{\theta},\pi_{0})) & = &  E \left\{ E \left( [s_{it}/\pi_{0,it}] \bm{X}_{it} \psi_{\tau} (Y_{it} - \bm{X}_{it}' \bm{\theta})  \;\middle\vert\; \bm{X}_{it}  \right) \right\} \\
& = & E \left\{ \bm{X}_{it} E \left( \psi_{\tau} (Y_{it} - \bm{X}_{it}' \bm{\theta}) E \left[ [s_{it}/\pi_{0,it}] \;\middle\vert\; \bm{W}_{it},\bm{X}_{it} \right]  \;\middle\vert\; \bm{X}_{it}  \right)  \right\} \\
           & = & E \left\{ \bm{X}_{it} E \left( \psi_{\tau} (Y_{it} - \bm{X}_{it}' \bm{\theta}) E \left[ [s_{it}/\pi_{0,it}] \;\middle\vert\; \bm{W}_{ih_i},\bm{X}_{it} \right]  \;\middle\vert\; \bm{X}_{it}  \right)  \right\} 
                                        %\\
\end{eqnarray*}
\begin{eqnarray*}
  E( M_{it}(\bm{\theta},\pi_{0}))        & = & E \left\{ \bm{X}_{it} E \left( \psi_{\tau} (Y_{it} - \bm{X}_{it}' \bm{\theta}) \;\middle\vert\; \bm{X}_{it}  \right)  \right\} = E \left\{ \bm{X}_{it} E \left( (\tau - I(Y_{it} \leq \bm{X}_{it}' \bm{\theta})) \;\middle\vert\; \bm{X}_{it}  \right)  \right\} \\
            & = & E \left\{ \bm{X}_{it} \left( \tau - F_{Y_{it}} \left(\bm{X}_{it}' \bm{\theta} \;\middle\vert\; \bm{X}_{it} \right) \right) \right\} = 0. 
\end{eqnarray*}
\end{proof}

\begin{proof}[Proof of Theorem \ref{T1}]
Let $\hat{\bm{\theta}} = (\hat{\bm{\vartheta}}',\hat{\bm{\alpha}}')' = \arg \min \{Q_{NT}(\bm{\theta},\hat{\bm{\gamma}}) \}$ and $\tilde{\bm{\theta}} = (\tilde{\bm{\vartheta}}',\tilde{\bm{\alpha}}')' = \arg \min \{Q_{NT}(\bm{\theta},\bm{\gamma_0}) \}$ where as before $\bm{\vartheta} = (\bm{\delta}',\bm{\beta}')'$ and $Q_{NT}(\bm{\theta},\pi(\bm{\gamma})) = Q_{NT}(\bm{\theta},\bm{\gamma})$ under Assumption \ref{A2}. We first show uniformly asymptotic equivalence of the objective functions, $Q_{NT}(\bm{\theta},\bm{\gamma}_0)$ and $Q_{NT}(\bm{\theta},\hat{\bm{\gamma}})$. We then show that the arguments that minimize the objective functions, $\hat{\bm{\theta}}$ and $\tilde{\bm{\theta}}$, are also asymptotically equivalent. Then we show that $\tilde{\bm{\theta}} \to \bm{\theta}_0$. 

Under Assumptions \ref{A2} and \ref{A5}, 
\begin{equation}
\sup_{\bm{\theta} \in \bm{\Theta}} | (Q_{NT}(\bm{\theta},\bm{\gamma}_0) - Q_{NT}(\bm{\theta}_0,\bm{\gamma}_0)) - (Q_{NT}(\bm{\theta},\hat{\bm{\gamma}}) - Q_{NT}(\bm{\theta}_0,\hat{\bm{\gamma}})) = o_p(1),
\end{equation}   
can be shown following Lemma 4 in Galvao, Lamarche and Lima (2013, Online Appendix). Moreover, noting that $Q_{NT}(\bm{\theta},\bm{\gamma})$ is a convex function in $\bm{\theta}$, the difference of the minimizers of the objective functions, $\hat{\bm{\theta}} - \tilde{\bm{\theta}} \to 0$ by Lemma \ref{L1}.

We now show that $\tilde{\bm{\theta}} \to \bm{\theta}_0$ following similar arguments to the one used in Kato, Galvao and Montes-Rojas (2012)'s Theorem 3.1. Let,
\begin{equation}
\MM_{N_i}(\bm{\theta}) = \MM_{N_i}(\bm{\delta},\bm{\beta},\bm{\alpha}) := \frac{1}{T} \sum_{t=1}^T \left( \omega_{it}(\bm{\gamma}_0) \rho_{\tau}( Y_{it} - \bm{d}_{it}'\bm{\delta} - \bm{x}_{it}'\bm{\beta} - \bm{z}_i' \bm{\alpha} ) + \lambda \rho_{\tau}(\bm{z}_i' \bm{\alpha}) \right),
\end{equation}
where $\omega_{it}(\bm{\gamma}_0) := s_{it}/\pi_{it}(\bm{\gamma}_0)$ and $\Delta_{N_i}(\bm{\theta}) = \MM_{N_i}(\bm{\theta}) - \MM_{N_i}(\bm{\theta}_0)$. For each $\eta > 0$, we define the ball $\mathcal{B}_i(\eta) := \{ (\bm{\delta}', \bm{\beta}', \alpha_{i}) : \| \bm{\delta} - \bm{\delta}_0 \|_1 + \| \bm{\beta} - \bm{\beta}_0 \|_1 + | \alpha_i - \alpha_{i0} | \leq \eta \}$ and the boundary $\partial \mathcal{B}_i(\eta) := \{ (\bm{\delta}', \bm{\beta}', \alpha_{i}) : \| \bm{\delta} - \bm{\delta}_0 \|_1 + \| \bm{\beta} - \bm{\beta}_0 \|_1 + | \alpha_i - \alpha_{i0} | = \eta \}$. For each $(\bm{\delta}', \bm{\beta}', \alpha_{i}) \not \in \mathcal{B}_i(\eta)$, define $\bar{\bm{\delta}}_i = r_i \bm{\delta} + (1 - r_i) \bm{\delta}_0$, $\bar{\bm{\beta}}_i = r_i \bm{\beta} + (1 - r_i) \bm{\beta}_0$, and $\bar{\alpha}_i = r_i \alpha_i + (1 - r_i) \alpha_{i0}$, where $r_i = \eta / (\| \bm{\delta} - \bm{\delta}_0 \|_1 + \| \bm{\beta} - \bm{\beta}_0 \|_1 + | \alpha_i - \alpha_{i0} |)$. Note that $r_i \in (0,1)$ and $\bar{\bm{\theta}}_i = (\bar{\bm{\delta}}', \bar{\bm{\beta}}', \bar{\alpha}_{i})'$ is in the boundary of $\mathcal{B}_i(\eta)$, $\partial \mathcal{B}_i(\eta)$. Because the convexity of the objective function holds for all $\lambda$ and therefore the objective function is convex, we have,
\begin{equation}
r_i \left( \MM_{N_i}(\bm{\theta}) - \MM_{N_i}(\bm{\theta}_0) \right) \geq \MM_{N_i}(\bar{\bm{\theta}}_i) - \MM_{N_i}(\bm{\theta}_0) = \EE(\Delta_{N_i}(\bar{\bm{\theta}}_i)) + \left( \MM_{N_i}(\bar{\bm{\theta}}_i) - \EE(\Delta_{N_i}(\bar{\bm{\theta}}_i)) \right).
\end{equation}

Note that $\EE(\Delta_{N_i}(\bar{\bm{\theta}}_i)) \geq \epsilon_\eta$ for all $ 1 \leq i \leq N$. As in Galvao et al. (2013)'s Theorem 1, we now need to show that for every $\epsilon > 0$, 
\begin{equation}
\max_{1 \leq i \leq N} P \left\{ \sup_{\bm{\theta} \in \mathcal{B}_i} | \Delta_{N_i}(\bm{\theta}) - \EE (\Delta_{N_i}(\bm{\theta})) | \geq \epsilon \right\} = o(N^{-1}). 
\end{equation}

Without loss of generality, we restrict all balls to be equal by setting $\alpha_{i0}=0$, $\bm{\beta}_0 = \bm{0}$ and $\bm{\delta}_0 = \bm{0}$. Thus, $\mathcal{B}_i(\eta) = \mathcal{B}(\eta)$ for all $ 1 \leq i \leq N$. Under Assumption \ref{A5}, following remark A.1 in Kato, Galvao and Montes-Rojas (2012), we observe that $| g_{\bm{\theta}}(u,a,\bm{X}) - g_{\bar{\bm{\theta}}}(u,a,\bm{X}) | \leq C(1+M) (\| \bm{\delta} - \bm{\delta}_0 \|_1 + \| \bm{\beta} - \bm{\beta}_0 \|_1 + | \alpha_i - \alpha_{i0} |)$, for some universal constant $C$ and $g_{\bm{\theta}}(u,a,\bm{x}) = (\rho_{\tau}(u - \bm{x}'\bm{\theta}) - \rho_{\tau}(u)) \omega(\bm{\gamma}_0) + \lambda (\rho_{\tau}(a - \bm{z}'\bm{\alpha}) - \rho_{\tau}(a))$. Since $\mathcal{B}(\eta)$ is a compact subset in $\RR^{p_x+p_d+1}$, $\exists$ $K$ $\ell_1$ balls with centers $\bm{\theta}^{(j)}$ for $j = 1,\hdots,K$ and radious $\epsilon/3 \kappa$, where $\kappa := C(1+M)$. For each $\bm{\theta} \in \mathcal{B}(\eta)$, there is $j \in \{1,\hdots,K\}$ such that $| g_{\bm{\theta}}(u,a,\bm{X}) - g_{\bm{\theta}^{(j)}}(u,a,\bm{X}) | \leq C(1+M) \epsilon/3 \kappa$, which leads to 
\begin{equation}
| \Delta_{N_i}(\bm{\theta}) - \EE ( \Delta_{N_i}(\bm{\theta}) ) | \leq | \Delta_{N_i}(\bm{\theta}^{(j)}) - \EE ( \Delta_{N_i}(\bm{\theta}^{(j)}) ) | + \frac{2 \epsilon}{3},
\end{equation}
and therefore,
\begin{eqnarray*}
P \left\{ \sup_{\bm{\theta} \in \mathcal{B}} | \Delta_{N_i}(\bm{\theta}) - \EE (\Delta_{N_i}(\bm{\theta})) | > \epsilon \right\} & \leq & P \left\{ \max_{1 \leq i \leq K} | \Delta_{N_i}(\bm{\theta}^{(j)}) - \EE (\Delta_{N_i}(\bm{\theta}^{(j)})) | + \frac{2 \epsilon}{3} > \epsilon \right\}  \\
& \leq & \sum_{j=1}^K P \left\{ | \Delta_{N_i}(\bm{\theta}^{(j)}) - \EE (\Delta_{N_i}(\bm{\theta}^{(j)})) | + \frac{2 \epsilon}{3} > \epsilon \right\}  \\
& = & \sum_{j=1}^K P \left\{ | \Delta_{N_i}(\bm{\theta}^{(j)}) - \EE (\Delta_{N_i}(\bm{\theta}^{(j)})) | > \epsilon/3 \right\}.
\end{eqnarray*}

By Hoeffding's inequality, each probability can be bounded by $2 \exp\left( - (\epsilon/3)^2 (T/2 M^2) \right)$, and therefore,
\begin{equation}
P \left\{ \sup_{\bm{\theta} \in \mathcal{B}} | \Delta_{N_i}(\bm{\theta}) - \EE (\Delta_{N_i}(\bm{\theta})) | \geq \epsilon \right\}  \leq 2 K \exp( - DT) = O(\exp(-T)),  
\end{equation}
where $D$ is a constant that depends on $\epsilon$. The desired result is obtained when $\log (N) / T \to 0$ as $N \to \infty$. 
\end{proof}

\begin{proof}[Proof of Theorem \ref{T2}]
The first part of the proof shows the weak convergence of the estimator using the arguments of Kato, Galvao and Montes-Rojas (2012)'s Theorem 3.2. We first obtain the Bahadur representation of $(\hat{\bm{\vartheta}} - \bm{\vartheta})$ and $(\hat{\alpha}_{i} - \alpha_{i0})$, then determine the rates of the reminder terms as in Kato et al., and finally obtain the asymptotic distribution after preliminary convergence rates were established. The reminder of the proof shows that the estimated $\omega_{it}(\bm{\gamma})$ does not affect the asymptotic distribution. 

Let $H_{N_i}^{(1)}(\bm{\theta}_i) := E(\mathbb{H}_{N_i}^{(1)}(\bm{\theta}_i))$ and $H_{N}^{(2)}(\bm{\theta}) := E(\mathbb{H}_{N}^{(2)}(\bm{\theta}))$ where the scores are: 
\begin{eqnarray*}
\mathbb{H}_{N_i}^{(1)}(\bm{\theta}_i) & := & \frac{1}{T} \sum_{i=1}^T \left( \frac{s_{it}}{\pi_{0,it}} \psi_{\tau}(Y_{it} - \bm{d}_{it}' \bm{\delta} - \bm{x}_{it}' \bm{\beta} - \alpha_{i}) - \frac{\lambda_T}{T} \psi_{\tau}( \alpha_{i} ) \right), \\
\mathbb{H}_{N}^{(2)}(\bm{\theta}) & := & \frac{1}{NT} \sum_{i=1}^N \sum_{i=1}^T  \frac{s_{it}}{\pi_{0,it}} \bm{V}_{it} \psi_{\tau}(Y_{it} - \bm{d}_{it}' \bm{\delta} - \bm{x}_{it}' \bm{\beta} - \bm{z}_i' \bm{\alpha}).
\end{eqnarray*}
where $\bm{\theta} = (\bm{\vartheta}',\bm{\alpha}')'$ and $\bm{\theta}_i = (\bm{\vartheta}', \alpha_i)'$. It follows then that,
\begin{eqnarray*}
H_{N_i}^{(1)}(\bm{\theta}_i) & = & E \left\{ \left( \tau - F_i( \bm{X}_{it}' (\bm{\theta} - \bm{\theta}_0) | \bm{X}_{it}) \right) [s_{it}/\pi_{0,it}] - (\lambda_T/T) ( \tau - G_i( \alpha_i - \alpha_{i0} ))  \right\} \\
H_{N}^{(2)}(\bm{\theta})   & = & \frac{1}{N} \sum_{i=1}^N E \left\{ \left( \tau - F_i( \bm{X}_{it}' (\bm{\theta} - \bm{\theta}_0) | \bm{X}_{it}) \right) [s_{it}/\pi_{0,it}] \bm{V}_{it} \right\}.
\end{eqnarray*}

The Bahadur representation of $(\hat{\bm{\vartheta}} - \bm{\vartheta}_0)$ and $(\hat{\alpha}_i - \alpha_{i0})$ can be obtained by expanding $H_{N_i}^{(1)}(\hat{\bm{\theta}}_i)$ and $H_{N}^{(2)}(\hat{\bm{\theta}})$ around $\bm{\theta}_0 = (\bm{\vartheta}_0',\bm{\alpha}_0')'$ and $\bm{\theta}_{i0} = (\bm{\vartheta}_0', \alpha_{i0})'$. We then obtain,
\begin{eqnarray}
H_{N_i}^{(1)}(\hat{\bm{\theta}}_i) & = & - \varphi_i (\hat{\alpha}_i - \alpha_{i0}) - \bm{E}_i' (\hat{\bm{\vartheta}} - \bm{\vartheta}_0) + O_p((\hat{\alpha}_i - \alpha_{i0}) ^2 \vee  (\hat{\bm{\vartheta}} - \bm{\vartheta}_0)^2 ), \label{eq:Hi1} \\
H_{N}^{(2)}(\hat{\bm{\theta}})   & = & - \frac{1}{N} \sum_{i=1}^N  \bm{E}_i (\hat{\alpha}_i - \alpha_{i0}) - \frac{1}{N} \sum_{i=1}^N \bm{J}_i (\hat{\bm{\vartheta}} - \bm{\vartheta}_0) + O_p((\hat{\alpha}_i - \alpha_{i0}) ^2 \vee  (\hat{\bm{\vartheta}} - \bm{\vartheta}_0)^2 ), \label{eq:H2}
\end{eqnarray}
where $\varphi_i := e_i - \lambda_T g_i/T$, $e_i := E(f_i(0|\bm{X}_{it}) [s_{it}/\pi_{0,it}])$, $\bm{E}_i := E(f_i(0|\bm{X}_{it}) [s_{it}/\pi_{0,it}] \bm{V}_{it})$, $g_i := E(g_i(0|\bm{X}_{it}) = E(g_i(0))$, and $\bm{J}_i := E(f_i(0|\bm{X}_{it}) [s_{it}/\pi_{0,it}]^2 \bm{V}_{it} \bm{V}_{it}')$. 

By the computational property of the quantile regression estimator (Gutenbrunner and Jureckova 1992\nocite{gutenbrunner1992}), $\max_{1 \leq i \leq N} | \mathbb{H}_{N_i}^{(1)}(\bm{\theta}_i) | = O_p(T^{-1})$. Then uniformly over $1 \leq i \leq N$, we have that,
\begin{equation*}
O_p(T^{-1}) =  \mathbb{H}_{N_i}^{(1)}(\bm{\theta}_{i0}) + H_{N_i}^{(1)}(\hat{\bm{\theta}}_i) + \{ \mathbb{H}_{N_i}^{(1)}(\hat{\bm{\theta}}_{i}) - H_{N_i}^{(1)}(\hat{\bm{\theta}}_i) + \mathbb{H}_{N_i}^{(1)}(\bm{\theta}_{i0}) \}.
\end{equation*}

Solving for $\hat{\alpha}_{i} - \alpha_{i0}$ in equation \eqref{eq:Hi1} gives,
\begin{eqnarray}
\hat{\alpha}_i - \alpha_{i0} & = & \varphi_i^{-1} \mathbb{H}_{N_i}^{(1)}(\bm{\theta}_{i0}) - \varphi_i^{-1} \bm{E}_i' (\hat{\bm{\vartheta}} - \bm{\vartheta}_0) + \varphi_i^{-1} \mathbb{H}_{N_i}^{(1)}(\hat{\bm{\theta}}_{i}) -  \varphi_i^{-1} H_{N_i}^{(1)}(\hat{\bm{\theta}}_i) \nonumber  \\
 &  & - \varphi_i^{-1} \mathbb{H}_{N_i}^{(1)}(\bm{\theta}_{i0}) +  O_p( T^{-1} \vee (\hat{\alpha}_i - \alpha_{i0}) ^2 \vee  (\hat{\bm{\vartheta}} - \bm{\vartheta}_0)^2 ) \label{eq:bahaduralpha}
\end{eqnarray}

Replacing equation \eqref{eq:bahaduralpha} in equation \eqref{eq:H2}, we obtain,
\begin{eqnarray*}
H_{N}^{(2)}(\hat{\bm{\theta}}) & = & - \frac{1}{N} \sum_{i=1}^N  (\bm{J}_i - \bm{E}_i \varphi_i^{-1} \bm{E}_i') (\hat{\bm{\vartheta}} - \bm{\vartheta}_0) - \frac{1}{N} \sum_{i=1}^N \varphi_i^{-1} \bm{E}_i \mathbb{H}_{N_i}^{(1)}(\bm{\theta}_{i0}) - \frac{1}{N} \sum_{i=1}^N \varphi_i^{-1} \bm{E}_i \bigg\{ \mathbb{H}_{N_i}^{(1)}(\hat{\bm{\theta}}_{i0}) \\
 &  &  - H_{N_i}^{(1)}(\hat{\bm{\theta}}_i) + \mathbb{H}_{N_i}^{(1)}(\bm{\theta}_{i0}) \bigg\} + o_p(\| \hat{\bm{\vartheta}} - \bm{\vartheta} \| ) + O_p( T^{-1} \vee \max_{1 \leq i \leq N} (\hat{\alpha}_i - \alpha_{i0}) ^2)
\end{eqnarray*}

Similarly than before, using the computational property of the quantile regression estimator, we obtain an expression for $H_{N}^{(2)}(\bm{\theta})$ and then solve for $\hat{\bm{\vartheta}} - \bm{\vartheta}$ obtaining,
\begin{eqnarray*}
\hat{\bm{\vartheta}} - \bm{\vartheta}_0 & = & \left( \frac{1}{N} \sum_{i=1}^N (\bm{J}_i - \bm{E}_i \varphi_i^{-1} \bm{E}_i') \right)^{-1} \left[ - \frac{1}{N} \sum_{i=1}^N \bm{E}_i \varphi_i^{-1} \mathbb{H}_{N_i}^{(1)}(\bm{\theta}_{0i}) + \mathbb{H}_{N}^{(2)}(\bm{\theta}_0) \right] \\
&  &  - \left( \frac{1}{N} \sum_{i=1}^N (\bm{J}_i - \bm{E}_i \varphi_i^{-1} \bm{E}_i') \right)^{-1}   \frac{1}{N} \sum_{i=1}^N \bm{E}_i \left\{ \mathbb{H}_{N_i}^{(1)}(\hat{\bm{\theta}}_{i}) - H_{N_i}^{(1)}(\hat{\bm{\theta}}_i) + \mathbb{H}_{N_i}^{(1)}(\bm{\theta}_{i0}) \right\}  \\
&  & + \left( \frac{1}{N} \sum_{i=1}^N (\bm{J}_i - \bm{E}_i \varphi_i^{-1} \bm{E}_i') \right)^{-1}   \left\{ \mathbb{H}_{N}^{(2)}(\hat{\bm{\theta}}) - H_{N}^{(2)}(\hat{\bm{\theta}}_i) - \mathbb{H}_{N}^{(2)}(\bm{\theta}_{0}) \right\}  + O_p( T^{-1} \vee \max_{1 \leq i \leq N} (\hat{\alpha}_i - \alpha_{i0}) ^2).
\end{eqnarray*}

By Theorem 3.2 in Kato, Galvao and Montes-Rojas (2012) (Steps 2 and 3), the first term is $O_p((NT)^{-1/2})$ and the other terms are asymptotically negligible under the conditions of the theorem. As $N^2 \log(N)^3 / T \to 0$, we obtain the Bahadur representation of the slope coefficient $\bm{\vartheta}$:
\begin{eqnarray*}
\sqrt{NT} (\hat{\bm{\vartheta}} - \bm{\vartheta}) & = & \left( \frac{1}{N} \sum_{i=1}^N \bm{J}_i - \bm{E}_i \varphi_i^{-1} \bm{E}_i' \right)^{-1} \left[ \sqrt{NT} \left( \mathbb{H}_{N}^{(2)}(\bm{\theta}_0) - \frac{1}{N} \sum_{i=1}^N \bm{E}_i \varphi_i^{-1} \mathbb{H}_{N_i}^{(1)}(\bm{\theta}_0) \right) \right]  + o_p(1) \\
                                                                                                                                                                                                  & = & \left( \frac{1}{N} \sum_{i=1}^N \bm{J}_i - \bm{E}_i \varphi_i^{-1} \bm{E}_i' \right)^{-1} \Bigg[  \frac{1}{\sqrt{NT}} \sum_{i=1}^N \sum_{t=1}^T  \left( \frac{s_{it}}{\pi_{0,it}}  \bm{V}_{it} - \bm{E}_i \varphi_i^{-1} \right) \psi_{\tau}(Y_{it} - \bm{X}_{it}' \bm{\theta}_0) \\
                                                                                                                                                                                                        &   & - \frac{\lambda_T}{\sqrt{T}} \frac{1}{\sqrt{N}} \sum_{i=1}^N \bm{E}_i \varphi_i^{-1} \psi_{\tau}(\alpha_{i0}) \Bigg]  + o_p(1).
\end{eqnarray*}

It follows that by the Liapunov Central Limit Theorem, under Assumption \ref{A8}, 
\begin{equation}
\sqrt{NT} (\hat{\bm{\vartheta}} - \bm{\vartheta}) \leadsto \mathcal{N}(\bm{0}, \bm{D}_1^{-1} \bm{D}_0 \bm{D}_1^{-1}). \label{AsyNorm}
\end{equation}

The result shown in \eqref{AsyNorm} holds for $\hat{\bm{\gamma}} = \bm{\gamma}_0$. It remains to show that $\omega_{it}(\hat{\gamma})$ does not affect the asymptotic distribution of the estimator. To this end, we use Lemma 4 in Galvao, Lamarche and Lima (2013) which is applied to show that $Q_{NT}(\bm{w},\hat{\bm{\gamma}}) - Q_{NT}(\bm{w},\bm{\gamma}_0) \to 0$. The result follows since,
\begin{eqnarray*}
\frac{1}{NT} \sum_{i=1}^N \sum_{t=1}^T | \omega_{it}(\hat{\bm{\gamma}}) - \omega_{it}(\bm{\gamma}_0) | & = & \frac{1}{NT}  \sum_{i=1}^N \sum_{t=1}^T s_{it} \left| \frac{\pi_{0,it} - \hat{\pi}_{it}}{\pi_{0,it} \cdotp \hat{\pi}_{it}} \right| \\
                                                                                                                                                                                                                                                                                                                                                                                                                & \leq & \frac{1}{NT}  \sum_{i=1}^N \sum_{t=1}^T s_{it} (\inf \pi_{0,it})^{-1} (\inf \hat{\pi}_{0,it})^{-1} \sup ({\pi_{0,it} - \hat{\pi}_{it}}) \\
                                                                                                      & \leq & O_p(1) \cdotp  O_p(1) \cdotp o_p(1),
\end{eqnarray*}
under Assumptions \ref{A2} and \ref{A3}.
\end{proof}

\end{document}